\documentclass[journal,9pt]{IEEEtran}
\usepackage{float}
\usepackage{comment}
\usepackage[table]{xcolor}
\usepackage{cite}      
\usepackage{graphicx}  
\usepackage{psfrag}    
\usepackage{hyperref}
\usepackage{placeins}
\usepackage{lineno}
\usepackage{url}
\usepackage{subcaption}
\usepackage{times}
\usepackage{textcomp}
\usepackage{amsmath} 
\usepackage{cleveref}
\usepackage{amsfonts,amsthm,amssymb}
\usepackage{balance}
\usepackage[normalem]{ulem}

\DeclareMathOperator{\E}{\mathbb{E}}

\DeclareMathOperator{\Pro}{\prod_{c \in \Phi}}
\DeclareMathOperator{\rcst}{\sigma_{t_{avg}}}
\DeclareMathOperator{\rcsc}{\sigma_{c_{avg}}}
\DeclareMathOperator{\Ei}{\underset{\sigma_c}{\E}}
\DeclareMathOperator{\Eii}{\underset{\sigma_c,G_c}{\E}}
\DeclareMathOperator{\Eiii}{\underset{\sigma_c,G_c,\Phi}{\E}}

\DeclareMathOperator{\Noise}{N_s}
\DeclareMathOperator{\hxt}{\mathcal{H}(r_t)}
\DeclareMathOperator{\hxc}{\mathcal{H}(R_c)}
\DeclareMathOperator{\hxcs}{\mathcal{H}(r_c)}

\newtheorem{theorem}{Theorem}[]
\newtheorem{corollary}{Corollary}[theorem]

\usepackage{multicol}
\usepackage{xr-hyper}

\newtheorem{definition}{Definition}

\begin{document}
%\linenumbers
\title{Optimization of Radar Parameters for Maximum Detection Probability Under Generalized Discrete Clutter Conditions Using Stochastic Geometry}
\author{Shobha~Sundar~Ram,~\IEEEmembership{Senior Member, IEEE}, Gaurav~Singh, and Gourab~Ghatak,~\IEEEmembership{Member, IEEE}
\thanks{S.S.R., G. S and G. G are with the Indraprastha Institute of Information Technology Delhi, New Delhi 110020 India. E-mail: \{shobha, gaurav18160, gourab.ghatak\}@iiitd.ac.in.}%
}
\maketitle

\begin{abstract}
We propose an analytical framework based on stochastic geometry (SG) formulations to estimate a radar's detection performance under generalized discrete clutter conditions. We model the spatial distribution of discrete clutter scatterers as a homogeneous Poisson point process and the radar cross-section of each extended scatterer as a random variable of the Weibull distribution. Using this framework, we derive a metric called the radar detection coverage probability as a function of radar parameters such as transmitted power, system noise temperature and radar bandwidth; and clutter parameters such as clutter density and mean clutter cross-section.
We derive the optimum radar bandwidth for maximizing this metric under noisy and cluttered conditions. We also derive the peak transmitted power beyond which there will be no discernible improvement in radar detection performance due to clutter limited conditions. When both transmitted power and bandwidth are fixed, we show how the detection threshold can be optimized for best performance.
We experimentally validate the SG results with a hybrid of Monte Carlo and full wave electromagnetic solver based simulations using finite difference time domain (FDTD) techniques. 
\end{abstract}

\providecommand{\keywords}[1]{\textbf{\textit{Keywords--}}#1}
%\IEEEpeerreviewmaketitle

\section{Introduction}
\label{sec:Intro}
Clutter models serve as predictive tools for setting detection thresholds and calibrating a radar's performance during operation \cite{skolnik1990radar}.
Empirical models of clutter from land terrains, sea and precipitation have been generated, over the decades, using large volumes of high quality measurement data \cite{long1975radar,4102550,barton1985land}. These data are typically laborious and time consuming to gather. Further, the computation of the statistical properties - such as higher order moments - from the data may be challenging \cite{lampropoulos1999high,ward2006sea}.  An alternate approach would be to use computational tools to model clutter returns \cite{rangaswamy1995computer, conte1987characterisation}. Full wave electromagnetic solvers are deterministic and computationally expensive especially at microwave and millimeter wave frequencies. Reliable characterization of radar detection metrics would require extensive trials of full wave simulations to capture the stochastic nature of the target and clutter conditions. In particular, we need to consider multiple target and clutter deployments and, for each of such realization, several clutter field-of-view and cross-section instances must be considered. Moreover, this procedure would have to be repeated for every set of possible radar parameters, e.g., bandwidth and center frequency. Therefore, such system level simulations are rarely used to quantify radar detection characteristics. In our paper, we propose the use of stochastic geometry (SG) to provide analytical expressions to quantitatively estimate a radar's detection performance.

In recent years, SG has been explored for a variety of applications~\cite{chiu2013stochastic}. For example, SG has been extensively studied in the context of different types of communication systems including cellular networks \cite{andrews2011tractable,bai2014coverage}, millimeter wave systems \cite{thornburg2016performance,ghatak2018coverage}, MIMO systems \cite{zia2007information,cui2012delay}, and vehicular networks \cite{beygi2015nested}. The framework has been exploited for optimizing communication parameters such as data rates under the given signal to interference ratio when multiple transmitters (such as base stations) and receivers (mobile end users) coexist in the channel. Such an analytical characterization is especially useful in capturing the diversity in the distribution and strength of the base station signals in the region of interest.

More recently, SG-based analysis has been extended to radar scenarios \cite{al2017stochastic,munari2018stochastic,ren2018performance,park2018analysis,fang2020stochastic}. In \cite{al2017stochastic,fang2020stochastic}, the authors modeled the distribution of the vehicles, mounted with automotive radars, on the roads as a Poisson point process, while in \cite{munari2018stochastic, park2018analysis}, the distribution of pulsed-radar sensors were modeled as a Poisson point process. In all of the above works, SG techniques are used to characterize the interference characteristics from multiple radars and analytically derive the signal to interference and noise ratio (SINR). The radar detection performance was then quantified on the basis of whether the SINR was above a pre-determined threshold. In this paper, we model the locations of discrete clutter scatterers as a Poisson point process similar to \cite{chen2012integrated}. Then, by employing tools from SG, we derive the signal to clutter and noise ratio (SCNR) in the region of interest.
The radar detection performance is subsequently quantified based on whether the SCNR exceeds a pre-determined threshold. Note that this formulation of the radar detection performance (and those of the prior SG-based radar works) is fundamentally distinct from classical radar detection frameworks, such as the Neyman-Pearson approach, where likelihood ratio tests between alternative and null hypotheses are performed to generate radar operating curves. In classical radar detection theory, the objective is usually to derive the probability density functions (pdf) for signal plus noise and clutter (or interference) - the  alternative hypothesis - and just noise and clutter (or interference)  - the null hypothesis \cite{kay1993fundamentals}. In contrast, our SG-based radar work, characterizes the pdf of the SCNR directly. The main advantage is that SG-based analysis characterizes the radar performance for all possible spatial configurations of network entities (in this case, discrete clutter scatterers) without the need for extensive simulations. Additionally, under suitable assumptions, such an analysis leads to tractable analytical results which enables the derivation of system-design insights that are often missed by system-level simulations~\cite{di2015computational}.

Two typical scenarios where discrete clutter are encountered are indoor radars \cite{amin2017radar} and foliage penetration radars \cite{davis2011foliage}. Both these radars typically operate at carrier frequencies below the X-band in order to facilitate non-line-of-sight (NLOS) spotting of humans. Indoor radars encounter significant discrete clutter scatterers such as furniture, walls, ceilings and floors \cite{bufler2016radariet}. Each of these discrete scatterers consist of multiple scattering centers with significant aspect angle variation. Many different algorithms and strategies have been proposed for mitigating both static and dynamic indoor clutter \cite{yoon2009spatial,solimene2013front,vishwakarma2017detection,vishwakarma2020mitigation}. All of these studies focus on specific room layouts and wall geometries. However, across different indoor environments, the scatterers may show considerable variations in their spatial distributions and sizes. Hence these radars may encounter considerable variations in the clutter returns during actual deployments. Similarly, in foliage penetration radar, returns from trees give rise to significant clutter \cite{davis2011foliage}. Again, there can be considerable spatial randomness in the tree's girth and distribution. In both the cases, the radar performance is based on the clutter characteristics such as its distribution, the radar cross-section (RCS) of each clutter scatterer as well as the interference between the scatterers. Similarly, variations in the target RCS and position can also affect the radar detection performance. The radar operating curves, in these scenarios, is directly a function of the SCNR which is a random function of both target and clutter parameters. In our preliminary work in \cite{ram2020estimating}, we proposed a metric termed as the radar detection coverage probability ($P_{DC}$) that approximates the fraction of the locations of a target in the region of interest where SCNR is above a predefined threshold. In that paper, we assumed an exponential distribution for the RCS of the discrete clutter and derived $P_{DC}$ for both line-of-sight (LOS) and NLOS conditions of the target. Radar clutter have been modeled using Rayleigh, log-normal or Weibull distributions in prior works \cite{goldstein1973false,schleher1976radar,conte2004statistical}. The log-normal model is typically used when the radar sees land clutter or sea clutter at low grazing angles, while the Rayleigh model is used when the amplitude probability distribution of the clutter is of a limited dynamic range. However, both of these models provide limiting distributions for most experimentally measured clutter. The Weibull model, on the other hand, provides a much more generalized model for radar clutter and is adopted in this work to model the returns from extended scatterers. We also assume that the noise, target and clutter statistics do not change appreciably during the coherent processing interval of the radar. These conditions are generally met for microwave or millimeter-wave radars. 
Our contributions in this paper are summarised below:
\begin{enumerate}
    \item We model discrete clutter scatterers as a homogeneous Poisson point process. We provide a comprehensive analysis of $P_{DC}$ under generalized clutter conditions modeled using the Weibull function. Here, the exponential and Rayleigh clutter models are just two specific cases of the generalized Weibull clutter model. $P_{DC}$ also considers diversity in the RCS of the radar target, the path loss function, the clutter density and the mean RCS of the clutter scatterers. The theorems, we offer, provide physics based insights into the radar performance with respect to different  radar, target and clutter parameters instead of  laborious measurement experiments and computationally complex simulation studies.
    \item The clutter returns that significantly affect the system performance arise from the same range cell occupied by the target. Greater bandwidth giving rise to small range cells result in weaker clutter returns.  Noise, on the other hand, increases proportionately to bandwidth. We provide a tractable method for optimizing the radar bandwidth for maximizing the radar detection coverage probability under discrete clutter and Gaussian noise conditions. The optimum radar bandwidth, derived from our theorem, will increase the likelihood of a target being detected across the entire region of interest.
    \item Under noise limited conditions, increase in transmitted power results in improved radar detection. However, in clutter limited scenarios, there is no further improvement in the radar's performance due to increased returns from the clutter scatterers. We provide an analytical solution for optimizing the radar transmitted power. When both the radar power and bandwidth are fixed, we provide a method to adjust the detection threshold.
    \item Finally, all of the existing SG works have validated the results using Monte Carlo simulations. In this work, we use a hybrid of finite difference time domain (FDTD) techniques based on full wave electromagnetic solvers and Monte Carlo simulations to validate our SG results. To the best of our knowledge, our work is the first to present an experimental validation of SG results with actual electromagnetic modeling.
\end{enumerate}
The paper is organized as follows. In Section \ref{sec:Theory}, we present two theorems. The first theorem derives the $P_{DC}$ using SG formulations while the second theorem derives the optimum radar bandwidth for maximizing $P_{DC}$. In the following Section \ref{sec:ExpValidation}, we describe the experimental validation using a hybrid of FDTD and Monte Carlo simulations. Finally, in Section \ref{sec:Results}, we present the key results and analyses regarding the radar detection performance as a function of radar parameters - such as radar bandwidth, transmitted power and receiver system noise - target parameters (target position, mean target RCS) and clutter parameters (mean clutter RCS, clutter density and type of distribution). The scientific inferences from our studies are summarized in the last section.

\emph{Notation} We use the following notation in the paper. We denote random variables by upper-case letters and their realizations or other deterministic quantities by lower-case letters. We use the bold font to denote vectors and normal font to denote scalar quantities.

\section{Theory}
\label{sec:Theory}
Consider a monostatic radar of wavelength $\lambda_c$ located at the origin of a two dimensional circular space and shown as a triangle in Fig. \ref{fig:ProblemSetup}. 
\begin{figure}
    \centering
    \includegraphics[width=2.5in,height=2.5in]{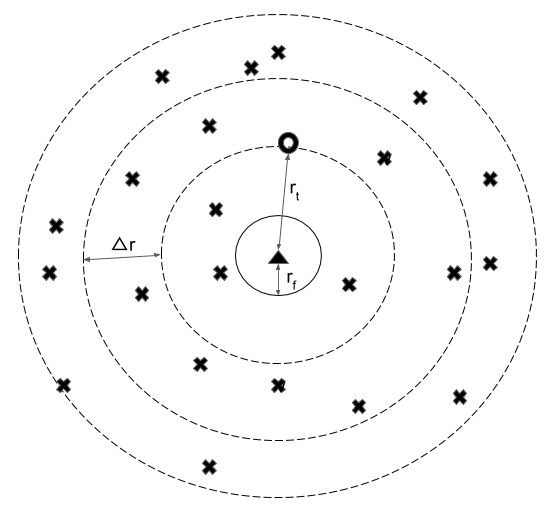}
    \caption{In the stochastic geometry formulation shown above, the monostatic radar (indicated as a triangle) is at the center of the region of interest with a single target (indicated as a circle) at a distance $r_t$ from the radar. Discrete clutter scatterers (indicated as crosses) are modelled as a homogeneous Poisson point process (PPP). They are distributed uniformly in the region of interest except within the far-field radius of the radar ($r_f$).}
    \label{fig:ProblemSetup}
\end{figure}
The transmitted power from the radar is $P_{tx}$. We assume that the radar consists of directional transmitting and receiving antennas of gain $G_{tx}(\theta)$ and $G_{rx}(\theta)$ respectively where $\theta$ is the direction-of-arrival of a direct path signal from a point scatterer (target or clutter). Noise at the radar receiver is assumed to be white Gaussian and is given by $\Noise = K_BT_sBW$ where $K_B$ is the Boltzmann constant, $T_s$ is the system noise temperature and $BW$ is the radar bandwidth. The bandwidth of the radar also determines the radar range resolution cell size given by $\Delta r = \frac{c}{2BW}$ where $c$ is the speed of light.

A single target, indicated as a circle in the figure, is assumed to be present in the channel at a fixed position $x_t = (r_t,\theta_t)$. Here, the target's Euclidean distance from the radar, $r_t$, can range from $r_t \in (r_f,\infty]$ where $r_f$ is the Fraunhofer far-field distance of the radar antennas.
The target azimuth, $\theta_t $, is fixed such that the target is assumed to be within the main lobe of the radar where $G(\theta_t)=G_{tx}(\theta_t)G_{rx}(\theta_t)=1$. The target radar cross-section ($\sigma_t$) is a random variable described by the Swerling 1 model with average RCS of $\rcst$ as shown in
\begin{align}
\label{eq:TargetRCS}
  P(\sigma_t) = \frac{1}{\rcst}exp \left(\frac{-\sigma_t}{\rcst} \right).
\end{align}
This corresponds to the case of a target composed of several scatterers of approximately equal reflectivities (like a human). 
Based on the radar range equation, the received signal at the radar for a single pulse is:
\begin{align}
\label{eq:RxPowSig}
    S(r_t) =  P_{tx}\sigma_t\hxt.
\end{align}
Here $\hxt$ is the path loss function which is a function of the distance between the radar and target as shown below
\begin{align}
\label{eq:PathLossFunc}
    \mathcal{H}(r_t)=PL_f\left(\frac{r_f}{r_t}\right)^{2q},
\end{align}
where $q$ is the path loss coefficient and $PL_f$ is the path loss factor at $r_f$.

Next, we discuss the clutter characteristics. We assume that there are discrete point scatterers that constitute the clutter (shown as crosses in the figure). We also assume that the coherent processing interval of the radar is short compared to the time required for the clutter statistics to change. These conditions are generally met for microwave or millimeter-wave radars. 
These clutter point scatterers are randomly distributed in the two-dimensional space. The locations of the random clutter scatterers are modeled as a homogeneous Poisson point process (PPP),  $\Phi$. The number of scatterers in a closed compact set $A$ is denoted by the intensity measure $\rho \nu(A)$ where $\rho$ is the clutter density and $\nu(A)$ is the size of the illuminated area, $A$.
Each realization of $\Phi$ is denoted by $\phi$ and the location of each clutter point is a random vector, $\mathbf{X_c}=(R_c,\theta_c)$. Again, the clutter points are assumed to be a minimum far-field distance $r_f$ away from the radar. The number of point clutter in each realization follows the Poisson's distribution while their distribution is assumed to be uniform within the region of interest from $[r_f,\infty)$ and $0 \leq \theta_c \leq 2\pi$. However, we only focus on the returns from those clutter points that lie within the same range resolution cell as the target (when $r_t-\frac{\Delta r}{2}<R_c<r_t+\frac{\Delta r}{2}$) since these are the returns that are most likely to impact the detection. Therefore the mean number of clutter scatterers within this range cell is given by $\rho 2 \pi r_t \Delta r$.
The radar signal reaching each clutter scatterer is affected by the path loss function based on a slow fading model denoted as $\hxc$. Each discrete clutter scatterer is considered an extended scatterer and modelled to have a fluctuating RCS ($\sigma_c$) based on the Weibull model with an average RCS of $\rcsc$, as shown in
\begin{align}
\label{eq:ClutterRCS}
 P(\sigma_c) = \frac{\alpha}{\rcsc}\left(\frac{\sigma_c}{\rcsc}\right)^{\alpha-1}\exp\left(-\left(\frac{\sigma_c}{\rcsc}\right)^{\alpha} \right),
\end{align}
where $\alpha$ is the shape parameter.
When $\alpha = 1$ the above expression reduces to the exponential probability distribution \eqref{eq:ClutterRCS}. When $\alpha = 2$, the above expression shows the Rayleigh probability distribution. 
Thus the total clutter returns at the radar receiver, for each realization of the PPP, depends on the clutter points within the range resolution cell of the target and is given by
\begin{align}
\label{eq:TotalClutter}
  C = \sum\limits_{c\in \phi, c \in r_t-\frac{\Delta r}{2}, r_t+ \frac{\Delta r}{2}}P_{tx} G(\theta_c)G_c\sigma_c \hxc.
\end{align}
In the above equation, the interference between the clutter returns from the individual point scatterers is captured by $G_c$. 

There can be significant variations in the clutter returns due to variation in the number, the distribution and fluctuations in RCS of the extended clutter scatterer. Similarly, the target returns may vary due to fluctuations in the target RCS.
Therefore, the \emph{mean} signal to clutter and noise ratio at the radar for a target at a given  $r_t$ is
\begin{align}
\label{eq:SCNR}
    SCNR(r_t) =\Eiii\left[\frac{P_{tx} \sigma_t\hxt}{\sum\limits_{c\in \phi}P_{tx} G(\theta_c)G_c\sigma_c \hxc +\Noise} \right].
\end{align}
Classical radar detection theory considers the radar operating curves derived from the probability of detection ($P_D$) and probability of false alarm ($P_{FA}$). There are many works in literature that have proposed approximations to the relationships between $P_D$, $P_{FA}$ and SNR. However, as the scenario becomes more complex, with a large number of discrete clutter scatterers with considerable variation in their spatial distribution and cross-sections, the relationship between $P_D$, $P_{FA}$ and SCNR becomes harder to derive analytically. Instead, we propose an alternative and simpler metric called the \emph{radar detection coverage probability} ($P_{DC}$) based on the mean SCNR.  Thus $P_{DC}$ is distinct from both $P_D$ and $P_{FA}$ used in classical radar and indirectly includes the effect of both detections and false alarms.
The metric is analogous to the \emph{cellular coverage probability} in wireless communications. There the metric is defined as the probability that a mobile user at any particular position in the coverage area will experience a signal to interference and noise ratio above a predefined threshold. The metric provides a method for evaluating the network performance while trading off between the benefits of increased capacity with greater density of mobile base stations with the performance issues that arise due to interference between the base stations. The metric also provides a method for optimizing some of the cellular parameters (such as data rate) for a given SINR. In the case of radar, greater bandwidth results in a smaller range resolution cell. If we consider the clutter that arises from the same cell as the target, then a smaller range cell results in reduced clutter. However, higher bandwidth also results in greater system noise at the radar receiver. Therefore, we propose to use the metric $P_{DC}$ to optimize the radar bandwidth with respect to noise and clutter conditions.

We first define the $P_{DC}$ metric below and provide a analytical framework for deriving it based on radar, target and clutter conditions.
\begin{definition}
The radar detection coverage probability ($P_{DC}$) is defined as the probability that the SCNR for a single target at a Euclidean distance $r_t$ from a monostatic radar, is above a predefined threshold $\gamma$: 
$P_{DC}(r_t) \triangleq \mathbb{P}\left(SCNR(r_t) \geq \gamma \right)$.
\end{definition}
\begin{theorem}
The radar detection coverage probability of a target at a distance $r_t$ from the radar is given by
\begin{align}
\label{eq:theorem1}
\begin{split}
P_{DC}(r_t) = exp \left( \frac{-\Noise \gamma}{P_{tx}\hxt \rcst}\right) \cdots \\
exp(-\rho r_t \Delta r \int_{\phi_c}(1-J(\theta_c))d\phi_c),
\end{split}
\end{align} 
where
\begin{align}
\label{eq:Jrc}
J(\theta_c)  
=\int_0^{\infty}exp\left(\frac{-\gamma G(\theta_c) \sigma_c}{ \rcst}\right)P(\sigma_c)d\sigma_c.
\end{align}
\end{theorem}
  
\begin{proof}
From \eqref{eq:SCNR}, we can see that the $P_{DC}$ at a given $r_t$ can be written in terms of the target cross-section in
\begin{align}
\label{eq:PdcGreaterThan}
    P_{DC}(r_t) = \mathbb{P}\left[\sigma_t >\frac{\Noise \gamma}{P_{tx}\hxt} +\frac{\gamma}{\hxt}\sum\limits_{c\in \phi} G(\theta_c)G_c\sigma_c \hxc \right].
\end{align}
Due to the exponential distribution of the target RCS in \eqref{eq:TargetRCS}, the above expression becomes
\begin{align}
\label{eq:Step1}
\begin{split}
\Eiii \left[ exp \left( \frac{-\Noise \gamma}{P_{tx}\hxt \rcst} +\sum\limits_{c\in \phi} \frac{-\gamma G(\theta_c)G_c\sigma_c \hxc}{\hxt \rcst}\right) \right] \\
=exp \left(\frac{-\Noise \gamma}{P_{tx}\hxt \rcst}\right) I(\cdot)
\end{split}
\end{align}
In the above expression, the terms within the first exponent - $N_s$, $\gamma$ $P_{tx}$, $\hxt$ and $\rcst$ - are all deterministic. This term encompasses the signal-to-noise (SNR) ratio with respect to the target returns. The second exponential term, $I(\cdot)$ shows the stochasticity of the clutter conditions (clutter cross-section and the PPP distribution of the clutter points). This term is independent of $P_{tx}$ and is a function of the signal-to-clutter ratio (SCR) of the system. $I(\cdot)$ can be further evaluated as
\begin{align}
\label{eq:Iexp}
I(\cdot) = \Eiii \left[ \Pro exp \left( \frac{-\gamma G(\theta_c)G_c\sigma_c \hxc}{\hxt \rcst}  \right) \right],
\end{align}
since the exponential of sum terms can be written as the product of exponential terms. Using the probability generating functional (PGFL) of SG \cite{haenggi2012stochastic}, we obtain
\begin{align}
\label{eq:PGFL}
\begin{split}
I(\cdot)= 
exp (-\rho \int_{r_t -\frac{\Delta r}{2}}^{r_t+\frac{\Delta r}{2}} \int_0^{2\pi} ( 1- \cdots \\
\Eii \left[exp\left(\frac{-\gamma G(\theta_c)G_c\sigma_c \hxcs}{\hxt \rcst}\right) \right] )d(\vec{x}_c) )\\ 
=
exp \left(-\rho \int_{r_t -\frac{\Delta r}{2}}^{r_t+\frac{\Delta r}{2}} \int_0^{2\pi} \left( 1- J(\cdot)\right) r_c d\phi_c dr_c \right)
\end{split}
\end{align}
As pointed out before, we only consider the discrete  clutter that arise within the same range cell as the target.
The inner expectation term, $J(\cdot)$, depends on the distribution of $\sigma_c$ and the interference term $G_c$. Under worst case scenarios, $G_c$ is always 1 and the clutter returns add. Then
\begin{align}
\label{eq:Jrc2}
\begin{split}
J(\cdot) = \Ei \left[exp\left(\frac{-\gamma G(\theta_c)\sigma_c \hxcs}{\hxt \rcst}\right) \right] \\
=\int_0^{\infty}exp\left(\frac{-\gamma G(\theta_c) \sigma_c \hxcs}{\hxt \rcst}\right)P(\sigma_c)d\sigma_c
\end{split}
\end{align}
Now if the range resolution cell is sufficiently narrow, which is usually the case for microwave and millimeter wave radars, then  $\hxcs \sim \hxt$ when $r_t -\frac{\Delta r}{2} \leq r_c \leq r_t + \frac{\Delta r}{2}$. Therefore \eqref{eq:Jrc2} becomes independent of $r_t$ as shown below in
\begin{align}
\label{eq:Jrc3}
J(\theta_c)  
=\int_0^{\infty}exp\left(\frac{-\gamma G(\theta_c) \sigma_c}{ \rcst}\right)P(\sigma_c)d\sigma_c.
\end{align}
Substituting this in \eqref{eq:PGFL}, we obtain
\begin{align}
\label{eq:SimpleI1}
    I = exp(-\rho r_t \Delta r \int_{\phi_c}(1-J(\theta_c))d\phi_c).
\end{align}
Combining \eqref{eq:Jrc2}, \eqref{eq:Iexp} and \eqref{eq:Step1}, we prove the theorem.
\end{proof}

In the above discussion not all distributions lead to tractable solutions. For example, the choice of exponential model (Swerling 1/2) model of the target RCS was crucial. The higher order Swerling 3 distribution results in far more challenging mathematical operations and hence not discussed here.  In \cite{ram2020estimating}, we discussed the effect of the gain of the radar antennas on the detection performance. We do not repeat that discussion here and confine our discussion to isotropic radar antennas where $G(\theta_c) =1$. The $J(\cdot)$ term which is a function of the clutter cross-section can be computed numerically. However for two cases, when $\alpha = 1$ corresponding to exponential distribution and for $\alpha = 2$ corresponding to Rayleigh distribution, analytical expressions for $J(\cdot)$ are derived. 

\emph{Case 1}: When $\alpha = 1$ for exponential distribution of clutter, with mean clutter cross-section, $\rcsc$, the $J(\cdot)$ term reduces to 
\begin{align}
\label{eq:JforAlpha1}
J = \frac{1}{1+\nu}     
\end{align}
where $\nu = \frac{\gamma \rcsc}{\rcst}$.
Substituting this expression back in \eqref{eq:Iexp}, we obtain
\begin{align}
\label{eq:IforAlpha1}
    I = exp\left( -2\pi \rho \frac{\nu}{\nu+1}r_t \Delta r \right)
\end{align}
which can be easily evaluated numerically.

\emph{Case 2:} When $\alpha = 2$ for Rayleigh distribution of clutter, the analytical solution for $J$ is
\begin{align}
\label{eq:JforAlpha2}
J = 1- \frac{\sqrt{\pi}\nu}{2} e^{\nu^2/4}erf(\frac{\nu}{2}),
\end{align}
where $erf(\cdot)$ is the error function. 
This results in 
\begin{align}
\label{eq:IforAlpha2}
    I = exp\left( -\pi^{1.5} \rho \nu e^{\nu^2/4}erf(\frac{\nu}{2}) r_t \Delta r\right)
\end{align}

\begin{corollary}
\label{corr:rcscConvergence1}
In case 1, when $\nu$ is much greater than 1 which corresponds to the situation when $\rcsc \gg \rcst$, then $I$ becomes independent of $\nu$. As a result, $P_{DC}$ becomes independent of $\rcsc$. 
Similarly, in case 2, the exponential term ($e^{\nu^2/4}$) within $J$ becomes very high when $\nu$ is high. As a result $I$ converges to 1 and the $P_{DC}$ is no longer a function of $\rcsc$.
In other words. $P_{DC}$ deteriorates with increase in $\rcsc$ till it asymptotically converges at a limit.
\end{corollary}
\begin{corollary}
\label{corr:rcscConvergence2}
The same effect of $P_{DC}$ versus $\rcsc$ is observed for generalized Weibull clutter parameter $\alpha$.
If we define $\kappa = \frac{-\gamma G(\theta_c) \hxcs}{\hxt}$ in the generalized $J$ in \eqref{eq:Jrc2}, then the first derivative of $J$ with respect to $\rcsc$ is given by
\begin{align}
\label{eq:dJbydSigma}
\begin{split}
    \frac{dJ}{d\rcsc} = \int_{0}^{\infty} exp(-\kappa \sigma_c) \cdots \\
    \left[ \frac{\alpha \sigma_c^{\alpha-1}}{\rcsc^{\alpha}}exp\left(-\left(\frac{\sigma_c}{\rcsc}\right)^{\alpha}\right) \left(\frac{-\alpha}{\rcsc^{\alpha+1}} -\alpha \left(\frac{\sigma_c}{\rcsc}\right)^{\alpha -1}\right) \right].
\end{split}
\end{align}
In \eqref{eq:dJbydSigma}, the terms inside the integral are always negative. Since $J$ is bounded below and a decreasing function of $\rcsc$, we can conclude that as $\rcsc$ tends to $\infty$, $J$ will tend to 0. Hence, for high values of $\rcsc$, $P_{DC}$ is independent of $\rcsc$ for generalized Weibull clutter conditions. 
\end{corollary}
\begin{corollary}
\label{corr:PtxConvergence}
In \eqref{eq:theorem1}, it is evident that the first exponential term indicates the effect of the SNR on the radar detection performance while the second term shows the effect of SCR. Hence increase in the transmitted power improves the radar detection performance while the radar operates in the noise limited scenario but has limited impact on the performance when the radar enters the clutter limited scenario. For fixed target and clutter conditions, $P_{DC}$ converges to  $I_{conv}$ with increase in $P_{tx}$. The transmitted power at which detection performance reaches 99\% ($e^{-0.01}$) of the convergence value is given by
\begin{align}
\label{eq:PtxEst}
\nonumber I_{conv}e^{-0.01} = exp \left( \frac{-\Noise \gamma}{P_{tx}^{max}\hxt \rcst}\right)
I_{conv}\\
    => P_{tx}^{max} = \frac{100\Noise \gamma}{\hxt \rcst}
\end{align}
The maximum transmitted power is therefore independent of clutter parameters such as $\rho$, $\rcsc$ and $\alpha$.
\end{corollary}
\begin{corollary}
\label{corr:NNdist}
In the above discussion, the radar detection coverage probability is provided in terms of the clutter density, $\rho$. However, in some experiments, other spatial statistical parameters besides clutter density may be used. One popular parameter is the average nearest neighbor distance, $\overline{r}$, which is the average distance between the centroid of a clutter point and its nearest neighboring clutter point. 
For a uniform random distribution of the clutter points, the average nearest neighbor distance is related to the clutter density \cite{ebdon1985statistics} as shown below
\begin{align}
\label{eq:roverbar}
    \overline{r} = \frac{1}{2\sqrt{\rho}}
\end{align}
Therefore, the radar detection probability in \eqref{eq:theorem1} can be written as 
\begin{align}
\label{eq:theorem2}
\begin{split}
P_{DC}(r_t) = exp \left( \frac{-\Noise \gamma}{P_{tx}\hxt \rcst}\right) \cdots \\
exp\left(-\frac{r_t \Delta r}{2\overline{r}^2} \int_{\phi_c}(1-J(\theta_c))d\phi_c\right)
\end{split}
\end{align} 
When the distribution of the clutter points deviate from the uniform distribution and cluster in some regions, the inhomogenity of the clutter distribution can be modeled either with the a non-uniform clutter density ($\rho(x_c)$) or through additional statistics on the average nearest neighbor distance \cite{ebdon1985statistics}.
\end{corollary}
Next, we provide a theorem for optimizing the radar $BW$ under noisy and cluttered conditions. 
\begin{theorem} For a narrow range resolution cell, the optimum $BW$ for detecting a Swerling 1 target is given by
\begin{align}
\label{eq:theorem3}
    BW = \sqrt{\frac{\pi \rho c  (1-J)r_t P_{tx} \hxt \rcst}{K_BT_s\gamma }}.
\end{align}
\end{theorem}
\begin{proof}
For a narrow $\Delta r$ and isotropic radar antennas, the integral in \eqref{eq:SimpleI1} can be reduced to
\begin{align}
\label{eq:SimpleI2}
    I(\cdot) = exp\left(-2 \pi \rho (1-J)r_t \Delta r \right)
\end{align}
Therefore the radar detection probability is a function of $BW$ as shown in
\begin{align}
\label{eq:BWfunc}
    P_{DC} = exp \left( \frac{-K_B T_s BW \gamma}{P_{tx}\hxt \rcst}\right) exp\left(-2 \pi \rho (1-J)r_t \frac{c}{2BW} \right).
\end{align}
Higher $BW$ results in greater noise which causes a reduction in the $P_{DC}$ due to the first exponential term. However increase in $BW$ also results in a reduced $\Delta r$ resulting in lesser clutter returns as seen in the second exponential term. Therefore, the optimum $BW$ for maximizing $P_{DC}$ can be determined by taking a natural logarithm on both sides of \eqref{eq:BWfunc} as shown below
\begin{align}
\label{eq:derivePdc}
    \ln{P_{DC}} =  \frac{-K_B T_s BW \gamma}{P_{tx}\hxt \rcst} -2 \pi \rho (1-J)r_t \frac{c}{2BW}. 
\end{align}
The optimum $BW$, shown in \eqref{eq:theorem3}, is obtained when the first derivative of the above expression with respect to $BW$ is equated to zero.
\end{proof}
The optimum $BW$ is shown to be a function of the target distance from the radar. In many situations it may not possible to change the radar bandwidth while tracking the target. In those scenarios, it may be preferable to be able to adjust the threshold $\gamma$ for a fixed transmitted power and radar bandwidth $BW$.
\begin{corollary}
\label{corr:optGamma}
The $\gamma$ for obtaining the maximum $P_{DC}$ for a given radar transmitted power and bandwidth for large clutter cross-sections should be adjusted based on the target distance $r_t$ using
\begin{align}
\label{eq:OptGamma}
\gamma (r_t) = \frac{\pi \rho c r_t P_{tx} \hxt \rcst}{K_BT_s BW^2}.
\end{align}
We have already observed from corollaries \ref{corr:rcscConvergence1} and \ref{corr:rcscConvergence2} that $J$ becomes 0 for high $\rcsc$. Therefore, in those scenarios, the above expression is directly obtained from \eqref{eq:theorem3}. Note that in the above expression, the $\gamma$ is independent of the type of clutter cross-section distribution ($\alpha$) and is only dependent on the density of the clutter scatterers.
\end{corollary}

\section{Experimental Validation}
\label{sec:ExpValidation}
In prior works, the experimental validation of the SG results were based on Monte Carlo simulations. However, the complete electromagnetic phenomenology (attenuation, diffraction, scattering) are not captured through these simulations. Therefore, in this paper, we use a full wave electromagnetic solver to model the complete radar propagation phenomenology. However, the computational complexity of these solvers is dependent on the size of the region of interest and the wavelength of the source excitation. Also, these solvers are inherently deterministic and cannot capture the diversity in radar, target and clutter parameters. Therefore, we use a hybrid of electromagnetic based modeling based on finite FDTD and Monte Carlo based simulations to experimentally validate the SG results. 
\subsection{Finite Difference Time Domain Simulations}
\label{sec:FDTD}
The FDTD technique models the complete propagation physics between a source and the scatterers in a medium. Many prior works have used FDTD for modeling indoor clutter \cite{bufler2016radar,ram2009radar, ram2010simulation, ram2016through, vishwakarma2020mitigation, vishwakarma2020micro}. For computational simplicity, we consider a two-dimensional (2D) simulation space along the $XY$ plane spanning 20m by 20m. A narrowband sinusoidal infinitely long, line source at 1 GHz emulating the monostatic radar is introduced at the center of the space at ($0,0$)m. The space is discretized into grid points that are uniformly spaced $\frac{\lambda_c}{10}$ apart where $\lambda_c$ is the wavelength corresponding to the source excitation.
The entire simulation space is bounded by a perfectly matched layer (PML) of $2\lambda_c$ thickness. 

The distribution of discrete clutter scatterers in the FDTD space follows the homogeneous PPP. For a given clutter density $\rho$, we perform multiple FDTD simulations as shown in Fig.\ref{fig:FDTD_realizations}. 
\begin{figure*}[ht]
    \centering
    \includegraphics[scale=1]{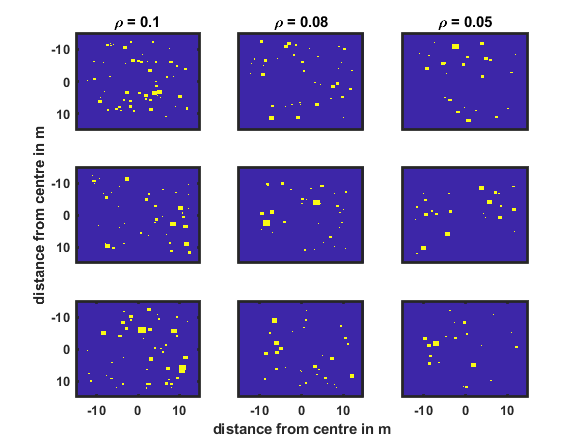}
    \caption{Two-dimensional FDTD simulation space ($20 \times 20$ m) with point clutter scatterers of $\epsilon_r = 7.1$ and mean RCS $\rcsc = 0.8 $ $m^2$ for three different clutter densities: first, second and third columns correspond to $\rho = 0.1, 0.08, 0.05$ $1/m^2$ respectively. The clutter scatterers are distributed uniformly in the simulation space while the RCS of the clutter scatterers follows the Weibull distribution with $\alpha =1$. }
    \label{fig:FDTD_realizations}
\end{figure*}
Each column in the figure shows three FDTD simulations for a specific $\rho$. The number of scatterers in a FDTD realization follows the Poisson distribution where the mean number of scatterers across multiple FDTD simulations is $\rho \times A$ where $A = 400$ $m^2$ is the area of the simulation space. In each FDTD simulation, the scatterers are distributed uniformly about the simulation space except within a radius of $r_f=3m$ around the source. Each scatterer is modelled as an infinitely long dielectric cylinder of dielectric constant $\epsilon_r = 7.1$ and radius $r_g$. 
The RCS, $\sigma_c$, of each of the clutter scatterers is a random variable drawn from the Weibull distribution of mean cross-section, $\rcsc = 0.8$ $m^2$ and shape parameter $\alpha = 1$ (or $2$) as given in \eqref{eq:ClutterRCS}. 
Using modal analysis, the 2D RCS of an infinitely long cylinder is given in terms of Bessel and Hankel functions as shown in 
\begin{equation}
\label{eq:CylinderRCS}
\sigma_c = \frac{16}{k}\left[-\frac{\mathcal{J}_0(kr_g)}{\mathcal{H}_0^1(kr_g)}+\sum_{n=1}^N 2(-1)^{n+1}\frac{\mathcal{J}_n(kr_g)}{\mathcal{H}_n^1(kr_g)}\right]^2,
\end{equation}
where $k=2\pi \sqrt{\epsilon_r}/\lambda$ and $N$ is the number of modes \cite{ruck1970radar}. The unit of the 2D RCS is in meters rather than square-meters (3D).
Based on the above equation, a look up table is formed between $r_g$ and $\sigma_c$ using a sufficiently large value of $N$ (50, in our case) when $\sigma_c$ has converged. Using this look up table, $r_g$ is estimated for each $\sigma_c$ of the scatterer. We further approximate the cylinders to have a square shaped longitudinal cross-section for simplicity. 

The FDTD models the time-domain transverse electric field, $E_z(\vec{r},t)$, in the two-dimensional grid space ($\vec{r}$). Through Fourier transform, we obtain the corresponding electric field in frequency domain, $E_z(\vec{r},f_c)$ for $f_c =$ 1 GHz. A second FDTD simulation is run in free space conditions without the presence of any of the scatterers using the same source excitation. Again, the resulting time-domain free space electric field is Fourier transformed to obtain the corresponding frequency-domain response at all the grid positions ($E_z^{fs}(\vec{r},f_c)$). Then, the two-way path loss, $\mathcal{H}(r)$, at a distance $r$ from the source is obtained from the ratio of the mean of the square of electric field at a distance $r$ from the source and the mean of the square of the corresponding electric field obtained in free space conditions, as shown in
\begin{align}
\label{eq:EstPathLoss}
    \mathcal{H}(r) =  \frac{\lambda_c}{(2\pi)^3 r^2}\left[\frac{ \oint_{\phi=0}^{2\pi}\left|E_z(\vec{r},f_c)^2\right|  d\phi }{\oint_{\phi=0}^{2\pi} \left|E_z^{fs}(\vec{r},f_c)^2 \right| d\phi}\right]^2.
\end{align}
The denominator term essentially normalizes the source excitation in the path loss factor. The power factor of 4 in \eqref{eq:EstPathLoss} accounts for the two-way propagation path and includes the effects of propagation through dielectric scatterers, diffraction about the edges of the scatterers and multipath reflections. 
Note that the path loss factor estimated from this simulation study corresponds to 2D cylindrical waves rather than 3D spherical waves that correspond to the usual radar scenario. Hence, the path loss decays at a rate of $1/m$ (2D) rather than $1/m^2$ (3D). However, in our SG formulations, we are only concerned with the relative power decay from $r_f$ to $r$ and hence the difference in the phase front in 2D and 3D scenarios is ignored. In \eqref{eq:EstPathLoss}, we estimated $\mathcal{H}(r)$ from the mean power decay corresponding to a circle of radius $r$ around the source. This is further averaged across multiple FDTD simulations to obtain a fairly good generalized estimate of $\mathcal{H}(r)$ for any given $\rho$ and $\rcsc$. 
We integrate these path loss estimates with Monte Carlo simulations to experimentally validate the stochastic geometry results.
\subsection{Monte Carlo Simulations}
\label{sec:MonteCarlo}
The SG results are validated through Monte Carlo simulations. The following parameters are fixed: $P_{tx} = 30$ dBm, $N_s =$ 300K, threshold $\gamma$ = 1 and carrier frequency of 1 GHz are fixed. The target cross-section, $\sigma_t$, in each trial is drawn from the Swerling 1 model in \eqref{eq:TargetRCS} with a mean of $\rcst = 0.8$ $m^2$ while the clutter cross-section of each scatterer is varied based on the Weibull model in \eqref{eq:ClutterRCS} with $\alpha = 2$ and $\rcsc = 0.8$ $m^2$. We consider three cases of clutter densities: $\rho = 0.1, 0.08, 0.05$ $/m^2$. For each case, the number of clutter scatterers in each trial is drawn from the Poisson distribution where the mean is $\rho \times A$ where $A$ is $400m^2$. We only consider the clutter returns from those that fall within the same range cell as the target. The size of the range cell is determined from the radar bandwidth $BW$. 
\begin{figure}[ht]
\begin{subfigure}{0.5\textwidth}
\centering
    \includegraphics[width=2in,height=2in]{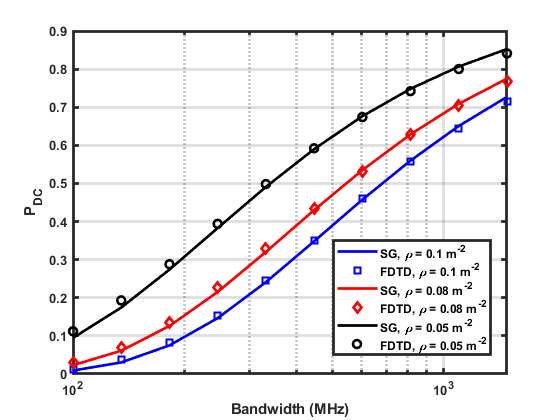}
    \caption{$P_{DC}$ vs. $BW, \rho$ for $\alpha = 1$}
    \label{fig:ExpRes_alpha1}
\end{subfigure}
\\
\begin{subfigure}{0.5\textwidth}
    \centering
    \includegraphics[width=2in,height=2in]{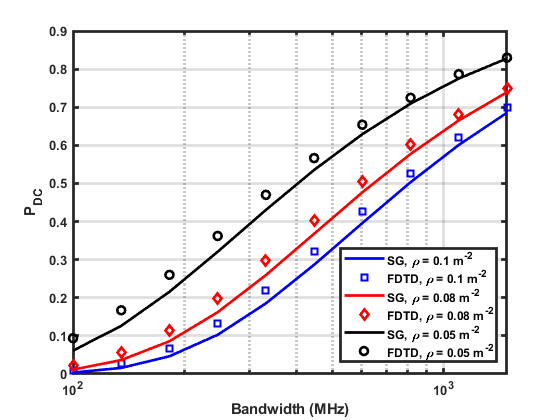}
    \caption{$P_{DC}$ vs. $BW, \rho$ for $\alpha = 2$}
    \label{fig:ExpRes_alpha2}
\end{subfigure}
\caption{FDTD and SG-based estimations of $P_{DC}$ for different $BW$ and $\rho$ when $P_{tx} = 30$ dBm, $N_s = 35$ K, $\rcst = \rcsc = 0.8$ $m^2$, $q=2$ for SG.}
\label{fig:FDTD_SG}
\end{figure}
The path loss factor for both target and clutter scatterers at any position $\vec{r}$ is obtained from the FDTD simulations described above. A total of 10000 trials are conducted to obtain $S$ and $C$ based on \eqref{eq:RxPowSig} and \eqref{eq:TotalClutter} respectively. The mean $P_{DC}$ is computed based on definition 1 and compared with the stochastic geometry solutions obtained from \eqref{eq:theorem1}. The $P_{DC}$ is plotted for different $BW$ and $\rho$. The results are presented in Fig.\ref{fig:FDTD_SG} for $\alpha=1$ and $\alpha=2$. The result shows an exact match between FDTD and SG results for $\alpha=1$ (Fig.\ref{fig:ExpRes_alpha1}) and a very close match when $\alpha=2$ (Fig. \ref{fig:ExpRes_alpha2}). In each figure, we show the values for the different cases of $\rho$. The results show that the SG framework is able to accurately estimate the radar's detection performance without the computational complexity of the FDTD framework. Note that the FDTD results have been confined to a region of $20 \times 20$m for a source of 1 GHz. The computation complexity would considerably increase for greater $r_t$ and higher source frequencies due to smaller grid dimensions.
The small variation in the results in Fig.\ref{fig:ExpRes_alpha2} may be attributed to the fact that the FDTD simulations do not have a circular cross-section as assumed. Second, we have assumed a path loss factor of $q=2$ in the SG formulations whereas the actual propagation conditions may have a slightly different $q$ due to multipath scattering. Third, the FDTD simulation is inherently deterministic. We estimate the \emph{mean} path loss factor by taking the average of a radius $r$ around the source. We also average across three simulations. However, the resulting estimate of the path loss factor may still not capture the stochasticity inherent in the radar propagation conditions. 
\section{Results}
\label{sec:Results}
In this section, we present the results obtained from the two theorems and their corollaries. The first theorem shows the relationship between the newly proposed metric - radar detection coverage probability - and radar, target and clutter parameters. 
\begin{table}[ht]
    \caption{Radar, Clutter and Target Parameters for Stochastic Geometry Formulations}
    \label{tab:parameters}
    \centering
    \begin{tabular}{c|c}
    \hline \hline 
    Radar Parameters & Values \\
    \hline
    Carrier frequency ($f_c$) & 1 GHz\\
    Transmitted power ($P_{tx}$) & 0 to 30 dBm\\
    Noise temperature ($N_s$) & 300 K to 1500 K\\
    Radar bandwidth ($BW$) &  0.1 to 2 GHz\\
    Antenna gain ($G$) & 0 dBi \\
    \hline
    Target Parameters & Values \\
    \hline
    Average target RCS ($\rcst$) & 1 $m^2$\\
    Target distance ($r_t$) & 5 to 50m\\
    \hline
    Clutter Parameters & Values \\
    \hline
    Average clutter cross-section ($\rcsc$) & 0.5 to 5 $m^2$ \\
    Weibull clutter  shape parameter ($\alpha$) &  1 to 2 (no units) \\
    Clutter density ($\rho$) & 0.003 $1/m^2$\\
    Path loss exponent ($q$) & 2 to 4 (no units)\\
    \hline
    SCNR threshold ($\gamma$) & 1 \\  
    \hline \hline
    \end{tabular}
\end{table}
The second theorem shows the optimum radar bandwidth under both noisy and cluttered conditions. In all of the cases, the radar antennas are assumed to be isotropic ($G_{rx}=G_{tx} = 0$) dBi with a minimum far-field distance of $r_f = 3$ m, while the carrier frequency is fixed at 1 GHz. The mean target cross section, $\rcst$, is assumed to be 1 $m^2$. All the radar parameters - $P_{tx}$, $N_s$ and $BW$ - can vary and the ranges of their variation are listed in Table.\ref{tab:parameters} along with target and clutter parameters.

\subsection{Analysis of Results from Theorem 1}
First, we present $P_{DC}$ as a function of target distance ($r_t$) from the radar for different types of clutter in Fig.\ref{fig:PdcVsRt_alphaparam}. 
\begin{figure*}[ht]
\begin{subfigure}{.3\textwidth}
\centering
    \includegraphics[width=2in,height=2in]{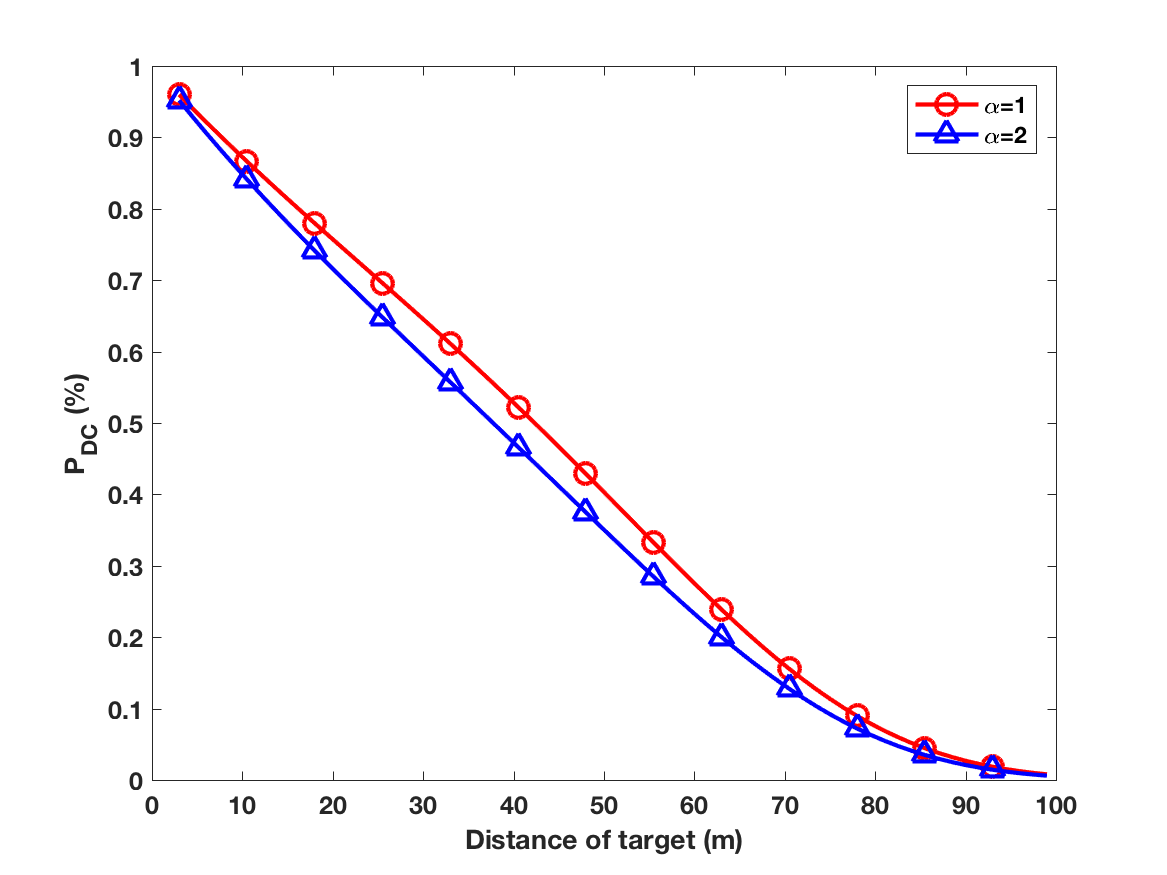}
    \caption{$P_{DC}$ vs. $r_t, \alpha$}
    \label{fig:PdcVsRt_alphaparam}
\end{subfigure}
\begin{subfigure}{.3\textwidth}
    \centering
    \includegraphics[width=2in,height=2in]{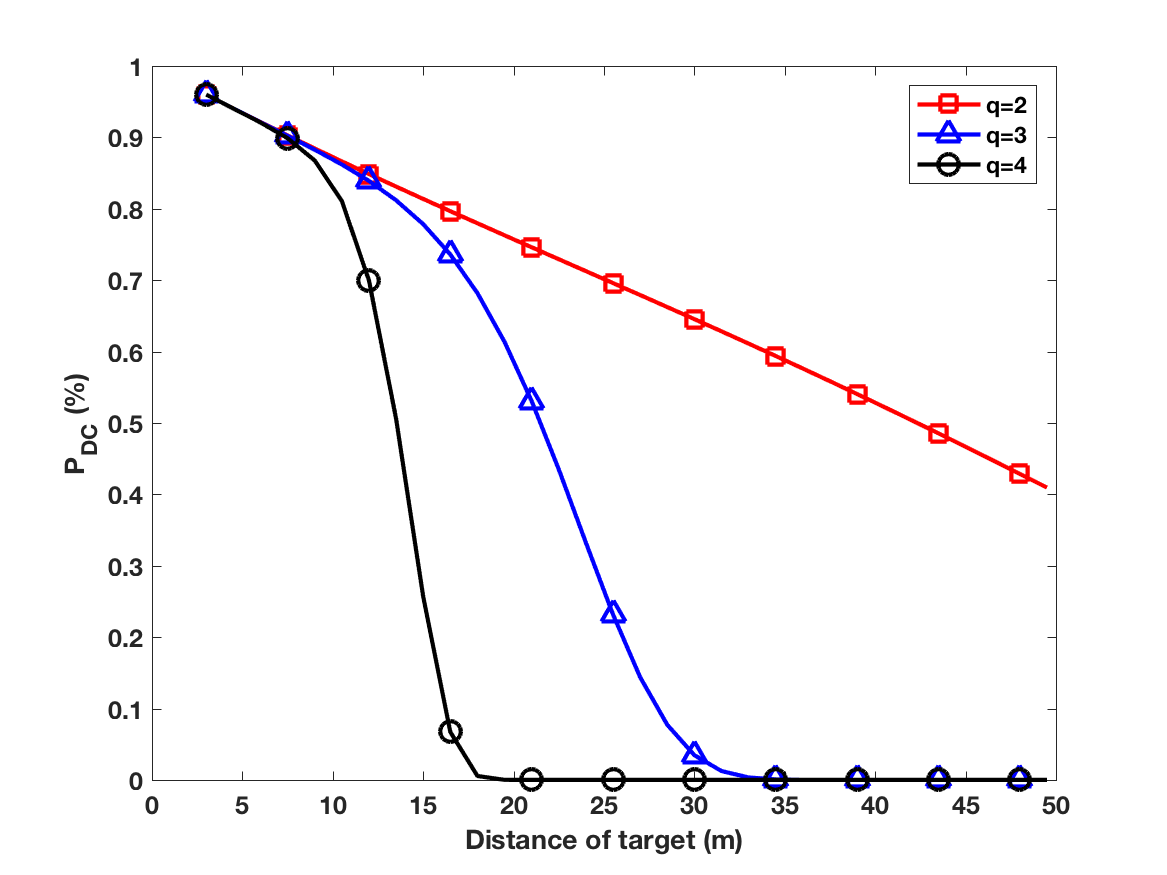}
    \caption{$P_{DC}$ vs. $r_t, q$, $\alpha = 1$}
    \label{fig:PdcVsRt_qparam}
\end{subfigure}
\begin{subfigure}{.3\textwidth}
    \centering
    \includegraphics[width=2in,height=2in]{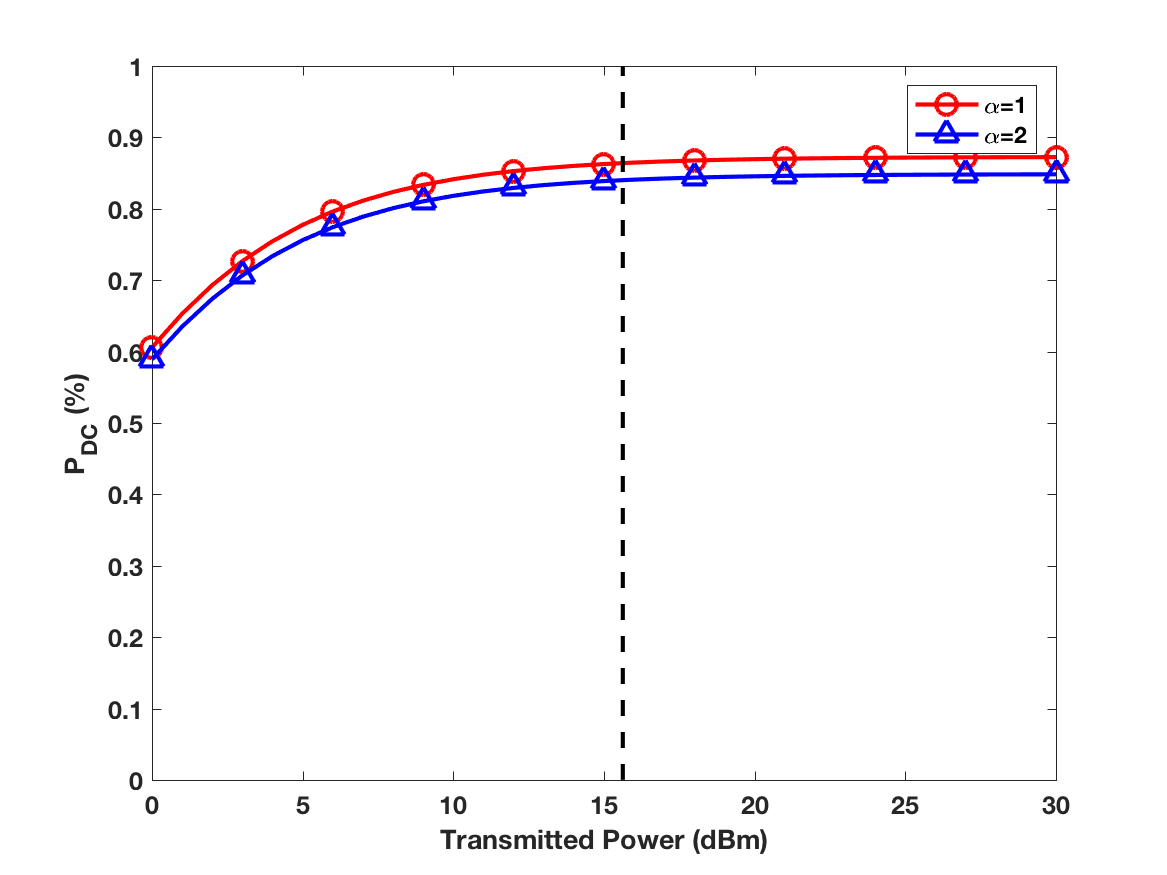}
    \caption{$P_{DC}$ vs. $P_{tx}, \alpha$}
    \label{fig:PdcVsPtx_alphaparam}
\end{subfigure}\\
\begin{subfigure}{.3\textwidth}
\centering
    \includegraphics[width=2in,height=2in]{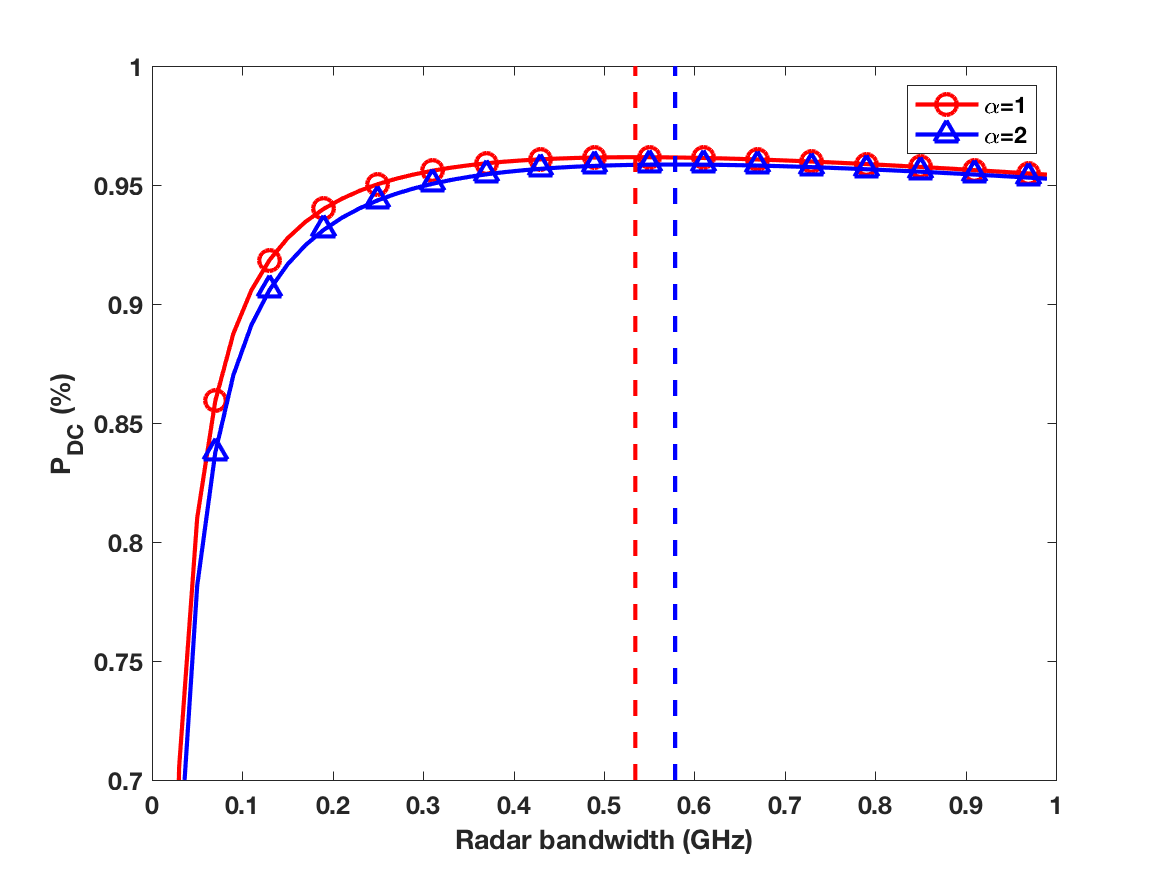}
    \caption{$P_{DC}$ vs. $BW,  \alpha$, $N_s = 1500$ K,}
    \label{fig:PdcVsBW_alphaparam}
\end{subfigure}
\begin{subfigure}{.3\textwidth}
    \centering
    \includegraphics[width=2in,height=2in]{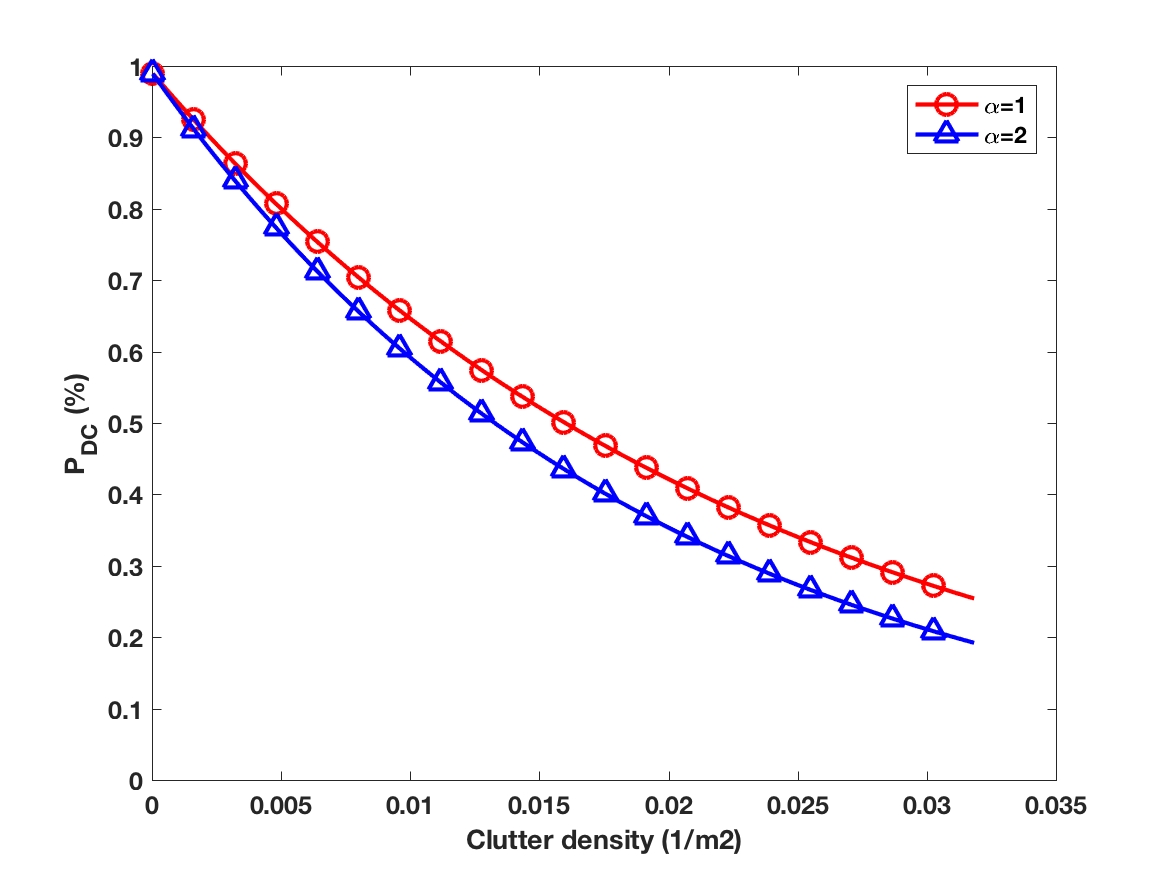}
    \caption{$P_{DC}$ vs. $\rho, \alpha$}
    \label{fig:PdcVsRho_alphaparam}
\end{subfigure}
\begin{subfigure}{.3\textwidth}
    \centering
    \includegraphics[width=2in,height=2in]{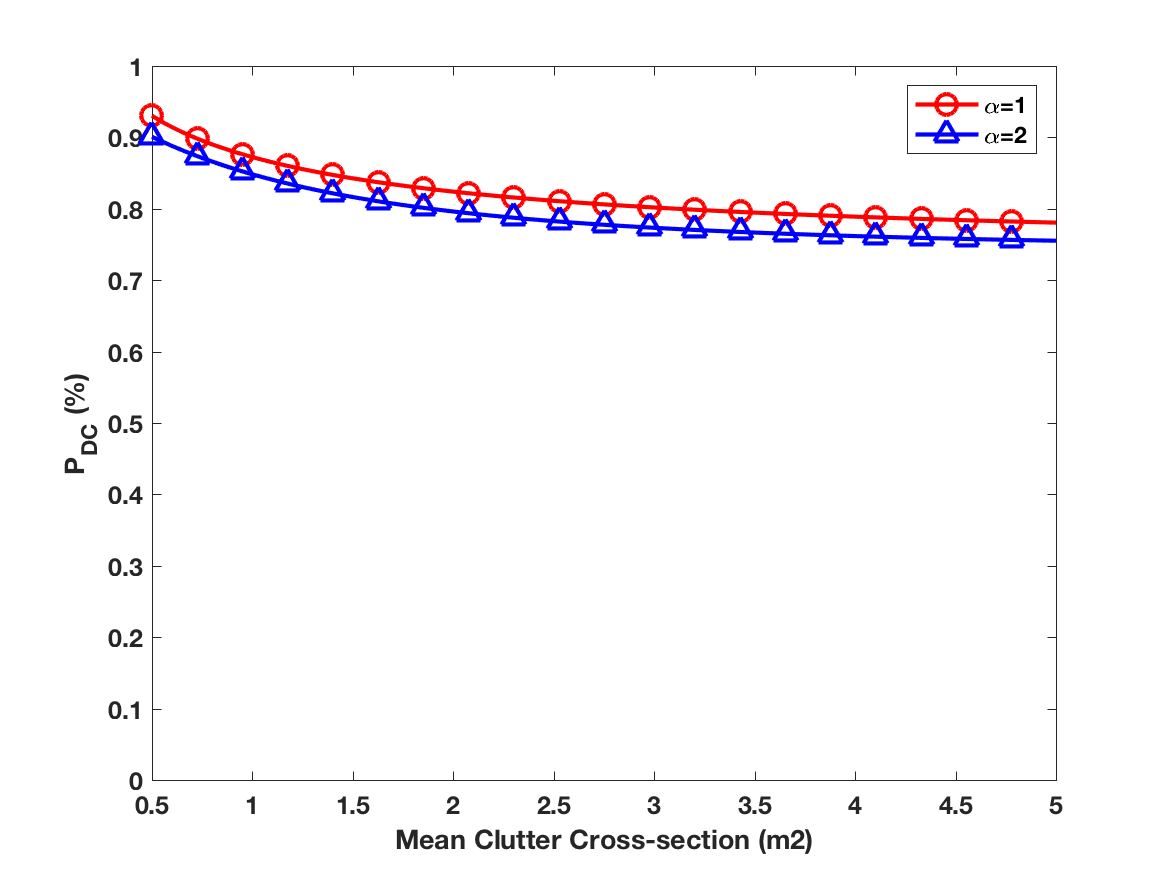}
    \caption{$P_{DC}$ vs. $\rcsc, \alpha$}
    \label{fig:PdcVssigmacavg_alphaparam}
\end{subfigure}
\caption{Theorem 1 results: Variation of $P_{DC}$ for different parameters when $r_f=3$ m, $\rcst = 1$ $m^2$, $\rho = 0.003$ $/m^2$. $P_{tx} = 30$ dBm for all figures except (c). $q = 2$ in all figures except (b). $r_t = 10$ m for figures (c)-(f). $BW = 0.1$ GHz, $N_s = 300$ K, in all figures except (d). $\rcsc = 1$ $m^2$ in all figures except (f).}
\label{fig:theorem1results}
\end{figure*}
For this case, $P_{tx}$ is 30 dBm, $N_s$ is 300 K and $BW$ is 0.1 GHz. Both $\rcsc = \rcst = 1 m^2$ while $q = 2$. We plot the variation of $P_{DC}$ for different values of Weibull shape parameter, $\alpha$. When $\alpha = 1$, the RCS of discrete clutter scatterers are of exponential distribution while $\alpha = 2$ corresponds to scatterers with RCS of Rayleigh distribution.  We observe that $P_{DC}$ falls with increase in $r_t$. This is because the increase in path loss term $\hxt$ results in the fall of the SNR (first exponential term) in \eqref{eq:theorem1}. The SCR in the second exponential term is less impacted by $\hxt$ since the clutter scatterers that fall within the same range cell as the target are similarly impacted by path loss. There is a slight increase in the clutter returns due to increase in clutter area size ($2 \pi r_t \Delta r$).
The detection performance deteriorates slightly for higher values of $\alpha$. The exponential distribution provides the upper bound of the performance while the Rayleigh distribution provides the lower bound. 

Next, we use the same parameters that were used in the previous case. However, this time $\alpha$ is fixed at 1 while $q$ is varied from 2 to 4 in Fig.\ref{fig:PdcVsRt_qparam}.
Higher values of $q$ incorporate multipath scattering effects into the path loss function $\mathcal{H}(\cdot)$. As expected, higher $q$ results in significant deterioration in the radar detection performance. 

Next, we examine the effect of increasing $P_{tx}$ on $P_{DC}$ in Fig.\ref{fig:PdcVsPtx_alphaparam} for different values of $\alpha$. Here, the distance of the target from the radar is fixed at $r_t =$10 m. $N_s$ and $BW$ are fixed at 300 K and 0.1 GHz respectively. Again, as before, $\rcsc = \rcst = 1 m^2$ while $q$ is 2. 
As mentioned earlier, the first term in \eqref{eq:theorem1} shows the effect of SNR on the radar detection performance. Here, as the $P_{tx}$ increases, the SNR improves resulting in higher detection performance. However, as pointed out in Corollary \ref{corr:PtxConvergence}, when we enter the clutter limited scenario, increase in $P_{tx}$ does not improve the performance of the radar since clutter returns proportionately increase. As a result, $P_{DC}$ converges. The dotted line in the figure shows the $P_{tx}$ value for which the $P_{DC}$ converges. This was derived from the analytical expression in corollary \ref{corr:PtxConvergence} in \eqref{eq:PtxEst}. Note that the maximum power is independent of $\alpha$ and other clutter related terms. 

Next, we run a similar study where we analyze the effect of the radar bandwidth on $P_{DC}$. The results are presented in Fig.\ref{fig:PdcVsBW_alphaparam} for different $\alpha$. In this scenario, $P_{tx}$ is 30 dBm, $N_s$ is 1500K while $r_t$ is 10m and $q$ is 2. All the other parameters are fixed as before except for $BW$ which is varied up to 1 GHz.
In a clutter limited scenario, higher $BW$ results in smaller range cells resulting in fewer clutter scatterers for a fixed clutter density.  Hence, we have higher signal-to-clutter (SCR) levels and $P_{DC}$. On the other hand, in a noise limited scenario, increase in $BW$ lowers the SNR and deteriorates the $P_{DC}$. Therefore, in the figure, we observe an optimum bandwidth where we obtain the peak $P_{DC}$. The two dotted lines indicate the optimum bandwidth derived from Theorem 2 \eqref{eq:theorem2}, for $\alpha = 1$ (red left line) and $\alpha = 2$ (blue right line). This result shows the correspondence between theorems 1 and 2. The result shows $P_{DC}$ improving with increase in BW (due to fall in noise) up to the optimum value after which the $P_{DC}$ falls slightly for higher values of $BW$ (due to fall in clutter returns). 

Figure.\ref{fig:PdcVsRho_alphaparam} shows the radar detection performance as a function of the clutter density ($\rho$) for a fixed target distance at $r_t=10m$. All the other radar parameters such as noise, bandwidth and transmitted power are fixed. We observe that, as expected, the increase in clutter density causes a fall in the SCR term in the second exponential term of \eqref{eq:theorem1}. 
Finally, we consider the effect of $\rcsc$ on $P_{DC}$ in Fig.\ref{fig:PdcVssigmacavg_alphaparam}. Based on the figure we observe that the $P_{DC}$ does fall up to a point with increase in $
\rcsc$. However, it asymptotically converges after a point. In other words, the increase in $\rcsc$ does not affect $P_{DC}$ beyond a point. This also concurs with the Corollaries \ref{corr:rcscConvergence1} and \ref{corr:rcscConvergence2}.

\subsection{Analysis of Results from Theorem 2}
Now, we discuss the results obtained from the second theorem based on Fig.\ref{fig:Theorem2results}. A radar's SNR falls with increase in radar bandwidth due to noise. But the SCR improves due to smaller range cells and fewer clutter scatterers. The optimum bandwidth is identified by Theorem 2 by considering both.  

First, we show how the optimum radar bandwidth derived from \eqref{eq:theorem2} varies as a function of target distance $r_t$ in Fig.\ref{fig:BWVsRt_alphaparam}. 
\begin{figure*}
\begin{subfigure}{.3\textwidth}
\centering
    \includegraphics[width=2in,height=2in]{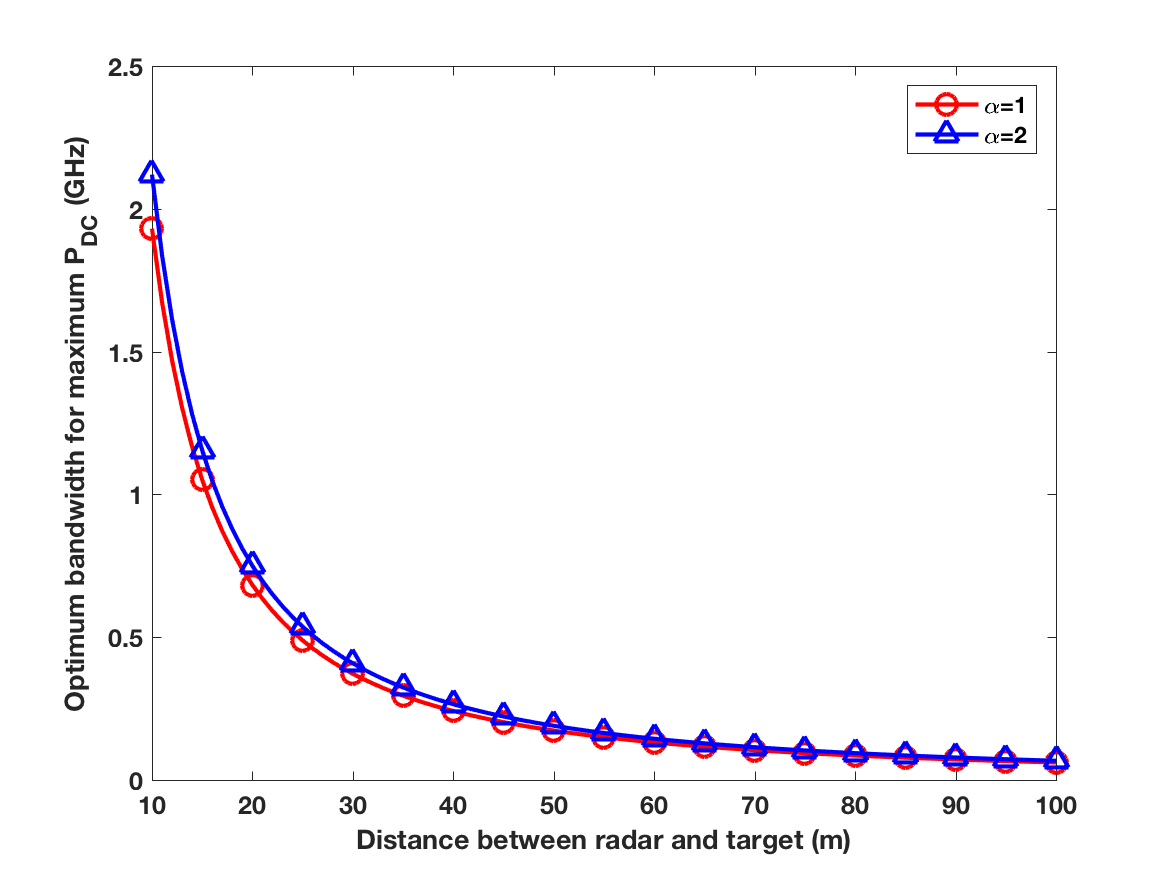}
    \caption{$BW$ vs. $r_t,\alpha$.}
    \label{fig:BWVsRt_alphaparam}
\end{subfigure}
\begin{subfigure}{.3\textwidth}
\centering
    \includegraphics[width=2in,height=2in]{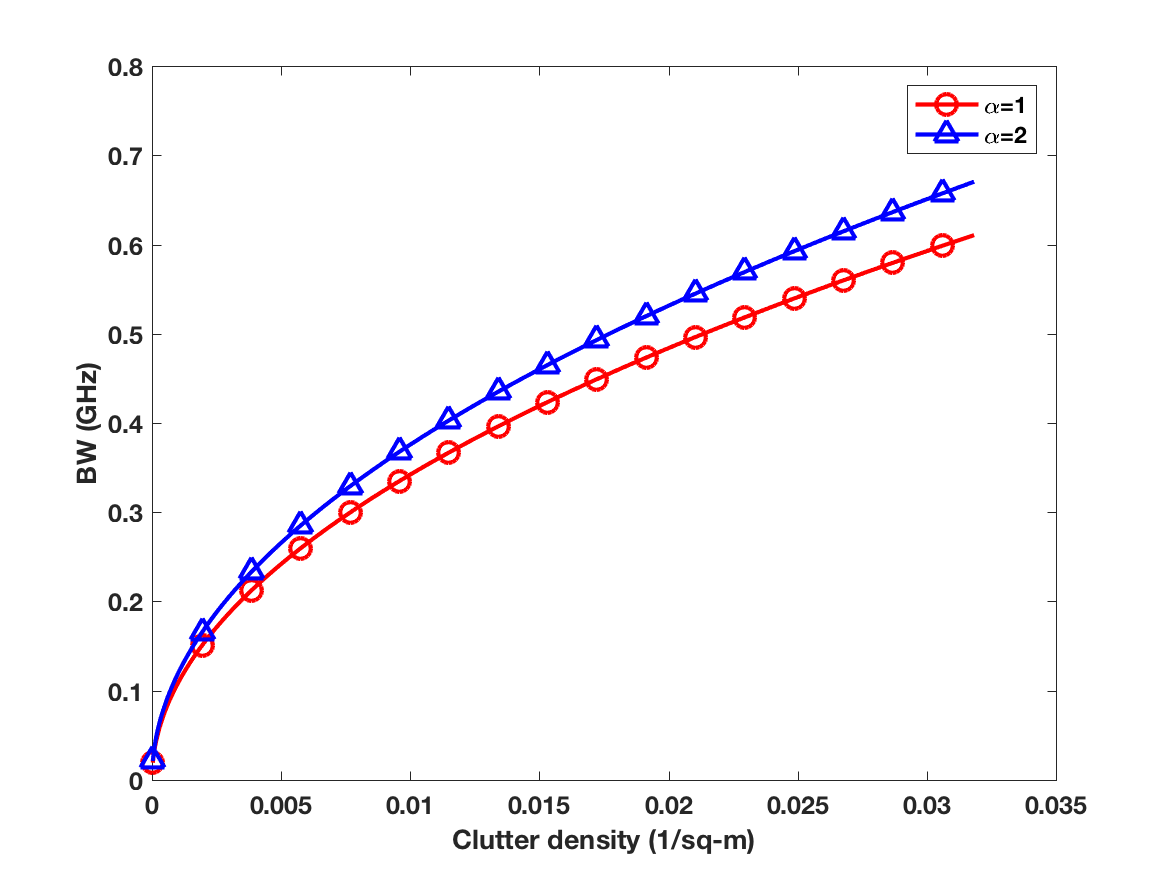}
    \caption{$BW$ vs. $\rho, \alpha$.}
    \label{fig:BWVsRho_alphaparam}
\end{subfigure}
\begin{subfigure}{.3\textwidth}
    \centering
    \includegraphics[width=2in,height=2in]{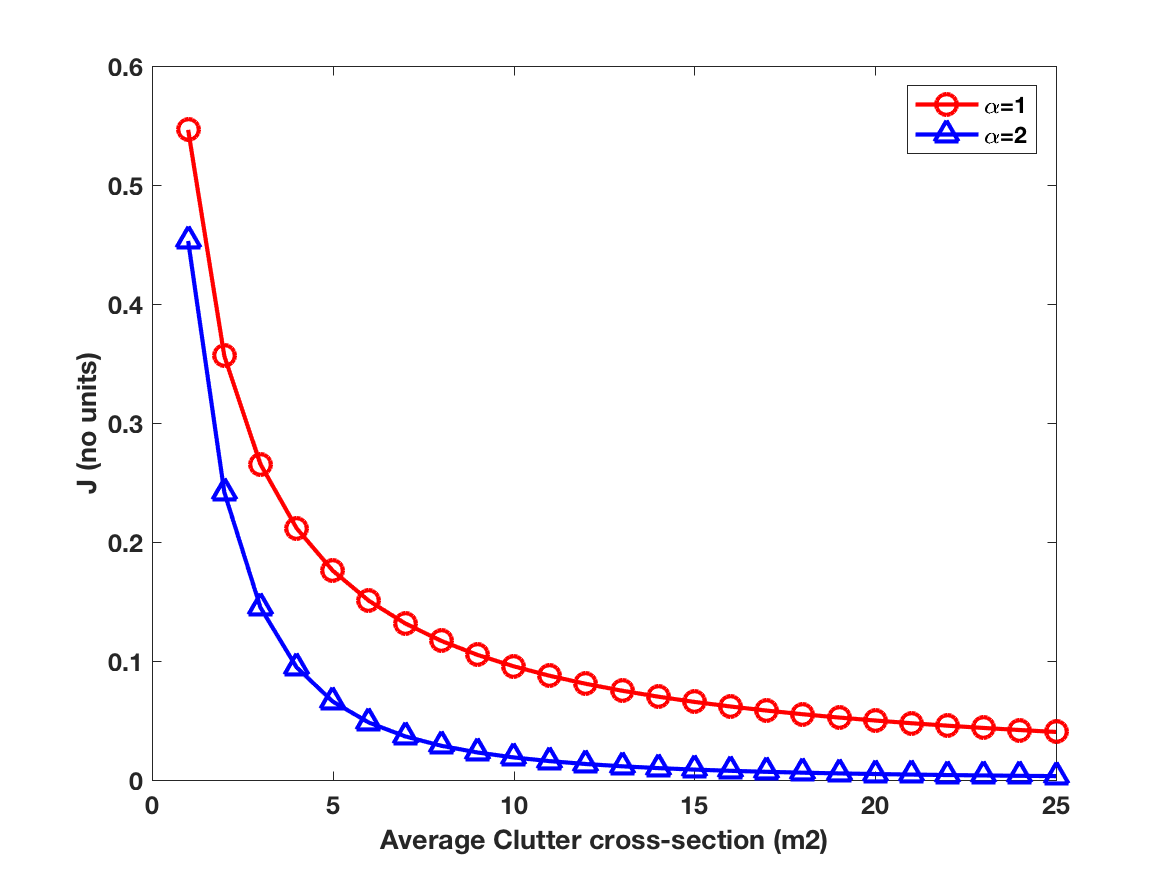}
    \caption{$J$ vs. $\rcsc, \alpha$}
    \label{fig:JVssigmacavg_alphaparam}
\end{subfigure}
\\
\begin{subfigure}{.3\textwidth}
    \centering
    \includegraphics[width=2in,height=2in]{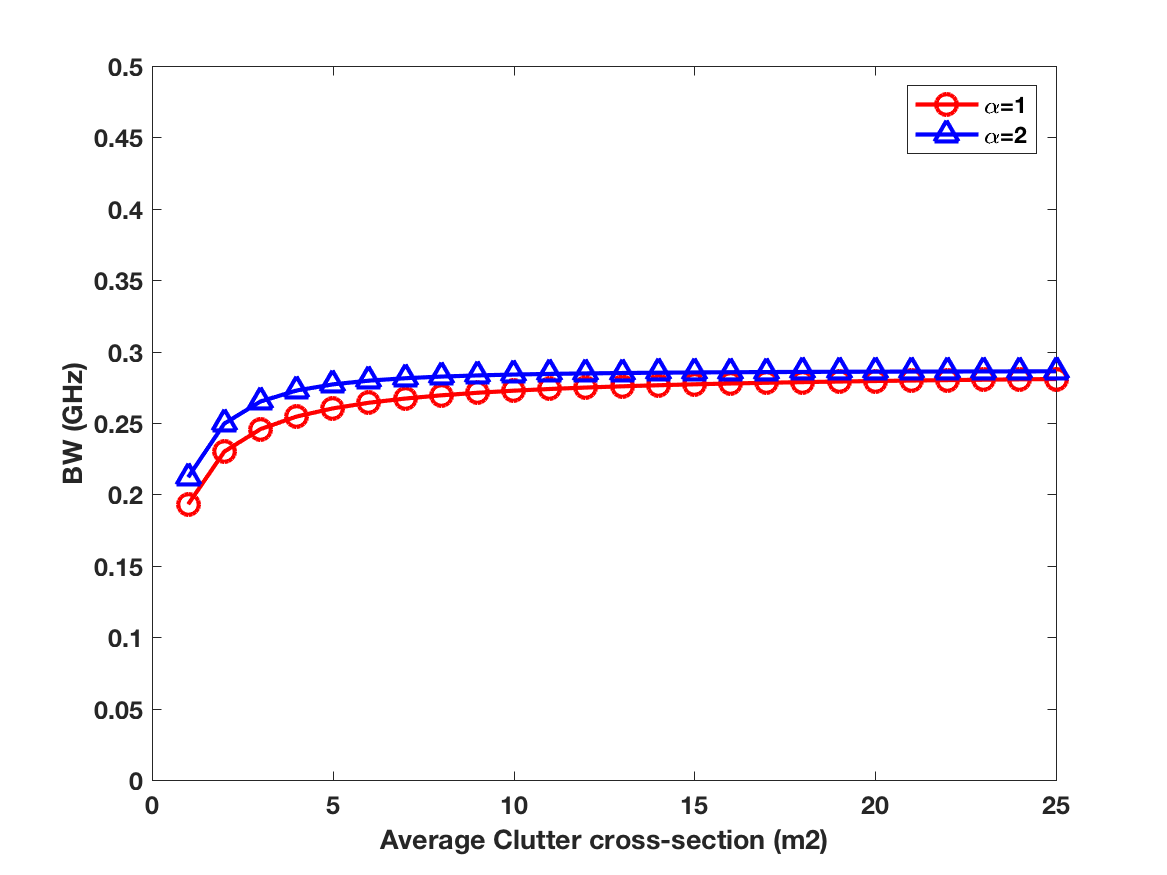}
    \caption{$BW$ versus $\rcsc, \alpha$}
    \label{fig:BWVssigmacavg_alphaparam}
\end{subfigure}
\begin{subfigure}{.3\textwidth}
    \centering
    \includegraphics[width=2in,height=2in]{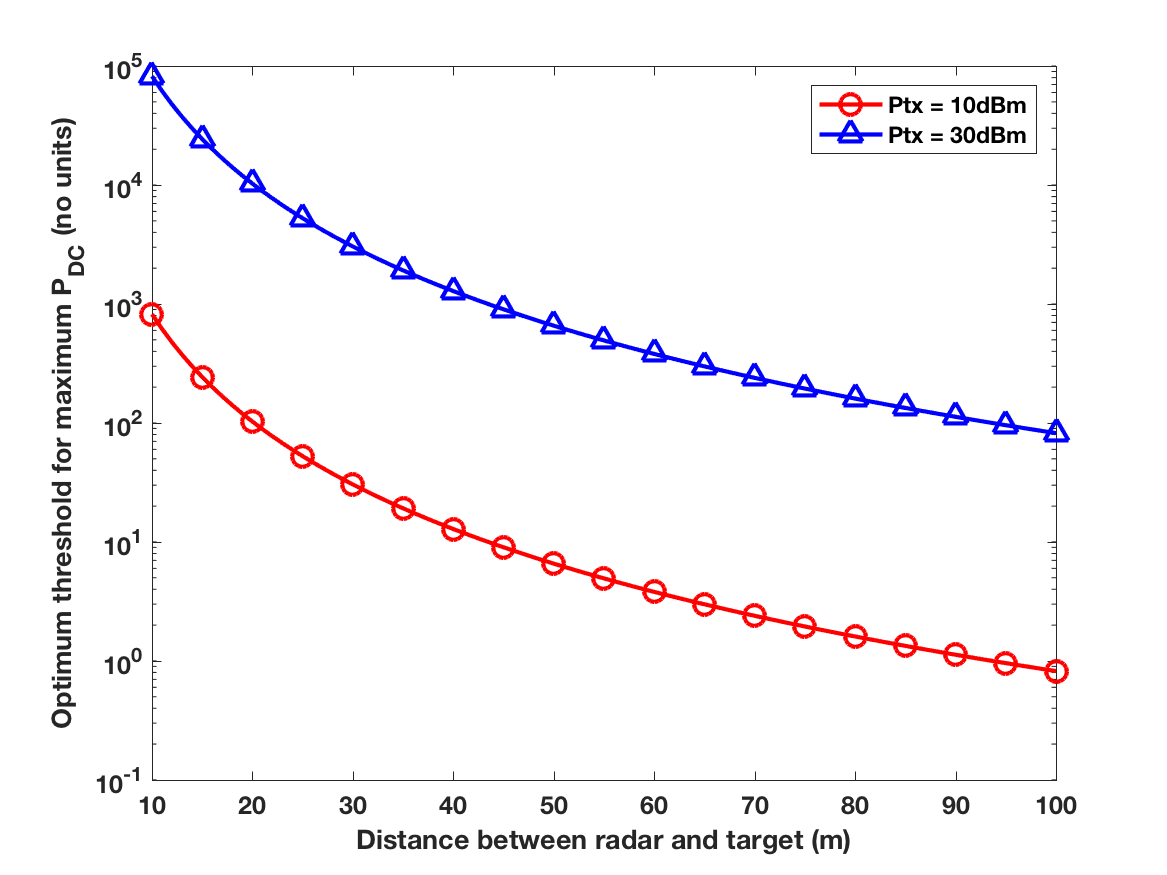}
    \caption{$\gamma$ vs. $r_t, P_{tx}$}
    \label{fig:gammaVsrt_Ptxparam}
\end{subfigure}
\begin{subfigure}{.3\textwidth}
    \centering
    \includegraphics[width=2in,height=2in]{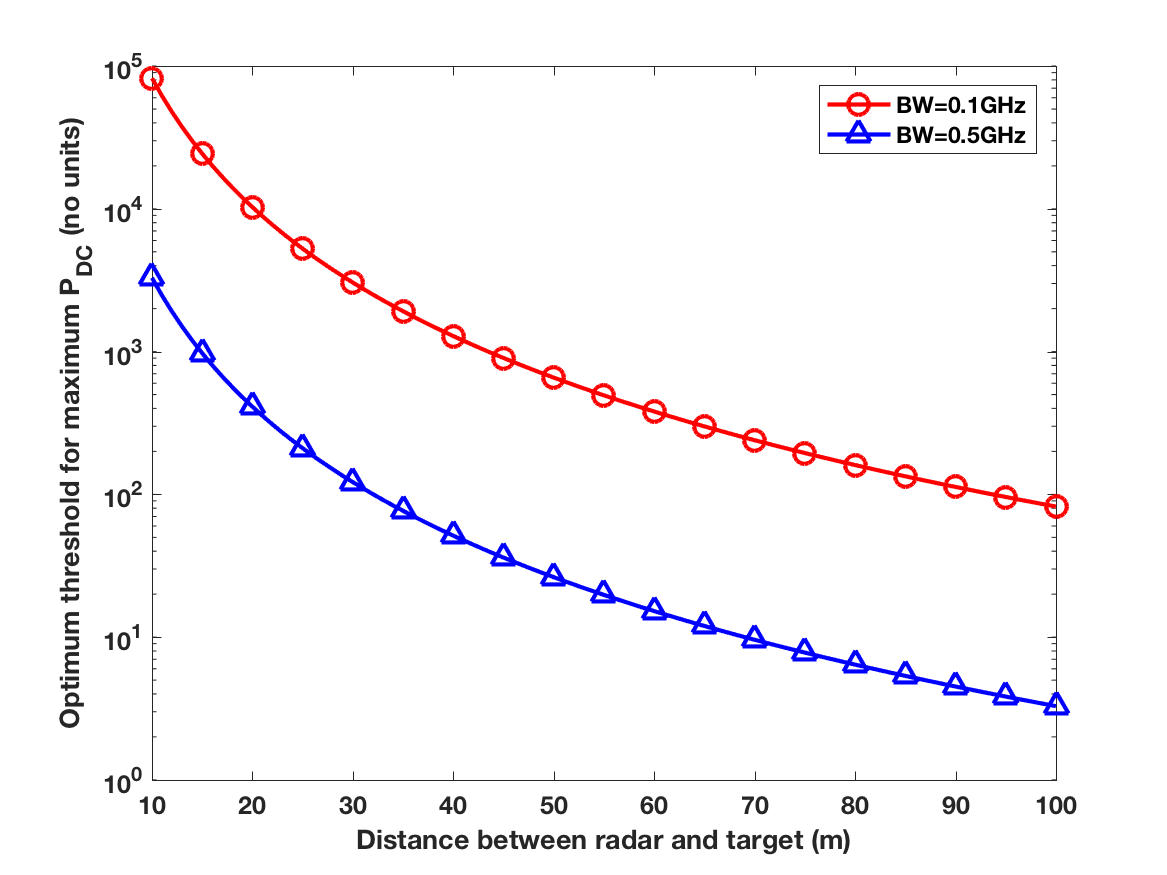}
    \caption{$\gamma$ versus $r_t, P_{tx}$}
    \label{fig:gammaVsrt_BWparam}
\end{subfigure}
\caption{Theorem 2 results: Optimum bandwidth for obtaining maximum $P_{DC}$ when $N_s = 300$ K, $q = 2$, $\rcst = 1$ $m^2$. $\rho$ fixed at $0.003$ $1/m^2$ in all cases except (b). $\gamma$ fixed at 1 in all cases except (e) and (f).}
\label{fig:Theorem2results}
\end{figure*}
Here, $P_{tx}$, $N_s$, $\rho$ and $q$ are fixed at 30 dBm, 300 K, 0.003 $1/m^2$ and 2 respectively. As we go farther, we observe that the optimum bandwidth reduces till it begins to converge. However, the variation in bandwidth is not considerable for high values of ranges. Therefore, it may be fairly easy to select an optimum bandwidth based on half the maximum unambiguous range of the radar field of view. The variation with respect to the clutter shape parameter is not very significant.

In Fig.\ref{fig:BWVsRho_alphaparam}, we observe that the BW is far more sensitive to the clutter density $\rho$. Here, all the other parameters except $\rho$ are fixed with the same values as the previous case. Higher clutter density requires a higher bandwidth or a smaller range cell size in order to balance the increase in noise. The returns are sensitive to the Weibull shape parameter. Slightly higher bandwidths are required for the Rayleigh distribution ($\alpha =2$) compared to the exponential distribution ($\alpha = 1$). 

In corollary \ref{corr:rcscConvergence1}, we observed that $J$ becomes negligible for higher values of $\rcsc$. This is also observed in Fig.\ref{fig:JVssigmacavg_alphaparam}. As a result, we observe how the optimum bandwidth for $\alpha =1$ and $\alpha = 2$ begin to converge to the same value in Fig.\ref{fig:BWVssigmacavg_alphaparam}. This supports the corollary \ref{corr:optGamma} where we noted that the optimum bandwidth becomes independent of the type of clutter distribution for high $\rcsc$. 

As a result, we can tune the threshold function $\gamma$ for a given transmitted power $P_{tx}$ and radar bandwidth $BW$ for obtaining the maximum $P_{DC}$ while tracking a target at a specific target distance of $r_t$. Fig.\ref{fig:gammaVsrt_Ptxparam} shows the ideal $\gamma$ as a function of $r_t$ for different $P_{tx}$ while keeping all other parameters fixed. We observe that $\gamma$ has to be reduced as we go farther from the radar to compensate for the path loss factor. Greater transmitted powers result in higher $\gamma$ as expected. Fig.\ref{fig:gammaVsrt_BWparam} shows how to tune $\gamma$ for different bandwidths. Higher bandwidths require lower $\gamma$. 

\section{Conclusion}
\label{sec:Conclusion}
We derived a metric - radar detection coverage probability ($P_{DC}$) - using SG techniques for evaluating the radar's detection performance under generalized discrete clutter conditions. The point clutter distribution were modeled as a homogeneous Poisson point process. The RCS of the target and clutter scatterers were modelled as random variables of Swerling-1 and Weibull distributions respectively. We evaluated the radar's performance for different radar, target and clutter parameters. 
%\balance
Our studies provide several insights into the performance of the radar under noisy and cluttered conditions. We list some of these below:
\begin{itemize}
\item The radar detection performance is best for the exponential distribution of the clutter cross-section (when Weibull shape parameter $\alpha = 1$) and worst for the Rayleigh distribution ($\alpha = 1$) for similar clutter densities and mean clutter RCS.
\item Increase in radar transmitted power improves the detection performance in noise limited scenarios. However, beyond a point there is no further discernible improvement since we enter the clutter limited scenario. We provide the analytical method to estimate the peak transmitted power at which the radar detection performance reaches its asymptotic maximum.
\item Large radar bandwidths result in increase in noise but lesser clutter returns due to smaller range cell size. We derive the optimum radar bandwidth for maximizing $P_{DC}$ under noisy and cluttered conditions. 
\item We show a method for optimizing the detection threshold for maximizing $P_{DC}$ for a fixed transmitted power and bandwidth. 
\end{itemize} 
Our results are experimentally validated with a hybrid of FDTD and Monte Carlo simulations. 
The FDTD simulations are used to model the path loss between the radar and a scatterer while the Monte Carlo simulations consider the diversity in the RCS of the target and clutter scatterers.  
\bibliographystyle{ieeetran}
\bibliography{main}

\begin{thebibliography}{10}
\providecommand{\url}[1]{#1}
\csname url@samestyle\endcsname
\providecommand{\newblock}{\relax}
\providecommand{\bibinfo}[2]{#2}
\providecommand{\BIBentrySTDinterwordspacing}{\spaceskip=0pt\relax}
\providecommand{\BIBentryALTinterwordstretchfactor}{4}
\providecommand{\BIBentryALTinterwordspacing}{\spaceskip=\fontdimen2\font plus
\BIBentryALTinterwordstretchfactor\fontdimen3\font minus
  \fontdimen4\font\relax}
\providecommand{\BIBforeignlanguage}[2]{{%
\expandafter\ifx\csname l@#1\endcsname\relax
\typeout{** WARNING: IEEEtran.bst: No hyphenation pattern has been}%
\typeout{** loaded for the language `#1'. Using the pattern for}%
\typeout{** the default language instead.}%
\else
\language=\csname l@#1\endcsname
\fi
#2}}
\providecommand{\BIBdecl}{\relax}
\BIBdecl

\bibitem{skolnik1990radar}
M.~I. Skolnik, ``Radar handbook second edition,'' \emph{McGrawHill}, 1990.

\bibitem{long1975radar}
M.~W. Long, ``Radar reflectivity of land and sea,'' \emph{Lexington, Mass., DC
  Heath and Co., 1975. 390 p.}, 1975.

\bibitem{4102550}
M.~Sekine, S.~Ohtani, T.~Musha, T.~Irabu, E.~Kiuchi, T.~Hagisawa, and
  Y.~Tomita, ``Weibull-distributed ground clutter,'' \emph{IEEE Transactions on
  Aerospace and Electronic Systems}, vol. AES-17, no.~4, pp. 596--598, 1981.

\bibitem{barton1985land}
D.~K. Barton, ``Land clutter models for radar design and analysis,''
  \emph{Proceedings of the IEEE}, vol.~73, no.~2, pp. 198--204, 1985.

\bibitem{lampropoulos1999high}
G.~Lampropoulos, A.~Drosopoulos, N.~Rey \emph{et~al.}, ``High resolution radar
  clutter statistics,'' \emph{IEEE Transactions on Aerospace and Electronic
  Systems}, vol.~35, no.~1, pp. 43--60, 1999.

\bibitem{ward2006sea}
K.~D. Ward, S.~Watts, and R.~J. Tough, \emph{Sea clutter: scattering, the K
  distribution and radar performance}.\hskip 1em plus 0.5em minus 0.4em\relax
  IET, 2006, vol.~20.

\bibitem{rangaswamy1995computer}
M.~Rangaswamy, D.~Weiner, and A.~Ozturk, ``Computer generation of correlated
  non-gaussian radar clutter,'' \emph{IEEE Transactions on Aerospace and
  Electronic Systems}, vol.~31, no.~1, pp. 106--116, 1995.

\bibitem{conte1987characterisation}
E.~Conte and M.~Longo, ``Characterisation of radar clutter as a spherically
  invariant random process,'' in \emph{IEE Proceedings F (Communications, Radar
  and Signal Processing)}, vol. 134, no.~2.\hskip 1em plus 0.5em minus
  0.4em\relax IET, 1987, pp. 191--197.

\bibitem{chiu2013stochastic}
S.~N. Chiu, D.~Stoyan, W.~S. Kendall, and J.~Mecke, \emph{Stochastic geometry
  and its applications}.\hskip 1em plus 0.5em minus 0.4em\relax John Wiley \&
  Sons, 2013.

\bibitem{andrews2011tractable}
J.~G. Andrews, F.~Baccelli, and R.~K. Ganti, ``A tractable approach to coverage
  and rate in cellular networks,'' \emph{IEEE Transactions on Communications},
  vol.~59, no.~11, pp. 3122--3134, 2011.

\bibitem{bai2014coverage}
T.~Bai and R.~W. Heath, ``Coverage and rate analysis for millimeter-wave
  cellular networks,'' \emph{IEEE Transactions on Wireless Communications},
  vol.~14, no.~2, pp. 1100--1114, 2014.

\bibitem{thornburg2016performance}
A.~Thornburg, T.~Bai, and R.~W. Heath, ``Performance analysis of outdoor mmwave
  ad hoc networks,'' \emph{IEEE Transactions on Signal Processing}, vol.~64,
  no.~15, pp. 4065--4079, 2016.

\bibitem{ghatak2018coverage}
G.~Ghatak, A.~De~Domenico, and M.~Coupechoux, ``Coverage analysis and load
  balancing in hetnets with millimeter wave multi-rat small cells,'' \emph{IEEE
  Transactions on Wireless Communications}, vol.~17, no.~5, pp. 3154--3169,
  2018.

\bibitem{zia2007information}
A.~Zia, J.~P. Reilly, J.~Manton, and S.~Shirani, ``An information geometric
  approach to ml estimation with incomplete data: application to semiblind mimo
  channel identification,'' \emph{IEEE Transactions on Signal Processing},
  vol.~55, no.~8, pp. 3975--3986, 2007.

\bibitem{cui2012delay}
Y.~Cui, V.~K. Lau, and Y.~Wu, ``Delay-aware bs discontinuous transmission
  control and user scheduling for energy harvesting downlink coordinated mimo
  systems,'' \emph{IEEE Transactions on Signal Processing}, vol.~60, no.~7, pp.
  3786--3795, 2012.

\bibitem{beygi2015nested}
S.~Beygi, U.~Mitra, and E.~G. Str{\"o}m, ``Nested sparse approximation:
  Structured estimation of v2v channels using geometry-based stochastic channel
  model,'' \emph{IEEE Transactions on Signal Processing}, vol.~63, no.~18, pp.
  4940--4955, 2015.

\bibitem{al2017stochastic}
A.~Al-Hourani, R.~J. Evans, S.~Kandeepan, B.~Moran, and H.~Eltom, ``Stochastic
  geometry methods for modeling automotive radar interference,'' \emph{IEEE
  Transactions on Intelligent Transportation Systems}, vol.~19, no.~2, pp.
  333--344, 2017.

\bibitem{munari2018stochastic}
A.~Munari, L.~Simi{\'c}, and M.~Petrova, ``Stochastic geometry interference
  analysis of radar network performance,'' \emph{IEEE Communications Letters},
  vol.~22, no.~11, pp. 2362--2365, 2018.

\bibitem{ren2018performance}
P.~Ren, A.~Munari, and M.~Petrova, ``Performance tradeoffs of joint
  radar-communication networks,'' \emph{IEEE Wireless Communications Letters},
  vol.~8, no.~1, pp. 165--168, 2018.

\bibitem{park2018analysis}
J.~Park and R.~W. Heath, ``Analysis of blockage sensing by radars in random
  cellular networks,'' \emph{IEEE Signal Processing Letters}, vol.~25, no.~11,
  pp. 1620--1624, 2018.

\bibitem{fang2020stochastic}
Z.~Fang, Z.~Wei, X.~Chen, H.~Wu, and Z.~Feng, ``Stochastic geometry for
  automotive radar interference with rcs characteristics,'' \emph{IEEE Wireless
  Communications Letters}, 2020.

\bibitem{chen2012integrated}
X.~Chen, R.~Tharmarasa, M.~Pelletier, and T.~Kirubarajan, ``Integrated clutter
  estimation and target tracking using poisson point processes,'' \emph{IEEE
  Transactions on Aerospace and Electronic Systems}, vol.~48, no.~2, pp.
  1210--1235, 2012.

\bibitem{kay1993fundamentals}
S.~M. Kay, \emph{Fundamentals of statistical signal processing}.\hskip 1em plus
  0.5em minus 0.4em\relax Prentice Hall PTR, 1993.

\bibitem{di2015computational}
M.~Di~Renzo, ``Computational stochastic geometry--on system-level modeling,
  simulation, performance evaluation, optimization, and experimental validation
  of 5g wireless communication networks,'' in \emph{2015 International
  Conference on Communications, Management and Telecommunications
  (ComManTel)}.\hskip 1em plus 0.5em minus 0.4em\relax IEEE, 2015, pp. 1--2.

\bibitem{amin2017radar}
M.~Amin, \emph{Radar for indoor monitoring: Detection, classification, and
  assessment}.\hskip 1em plus 0.5em minus 0.4em\relax CRC Press, 2017.

\bibitem{davis2011foliage}
M.~E. Davis \emph{et~al.}, \emph{Foliage penetration radar}.\hskip 1em plus
  0.5em minus 0.4em\relax SciTech Pub., 2011.

\bibitem{bufler2016radariet}
T.~D. Bufler and R.~M. Narayanan, ``Radar classification of indoor targets
  using support vector machines,'' \emph{IET Radar, Sonar \& Navigation},
  vol.~10, no.~8, pp. 1468--1476, 2016.

\bibitem{yoon2009spatial}
Y.-S. Yoon and M.~G. Amin, ``Spatial filtering for wall-clutter mitigation in
  through-the-wall radar imaging,'' \emph{IEEE Transactions on Geoscience and
  Remote Sensing}, vol.~47, no.~9, pp. 3192--3208, 2009.

\bibitem{solimene2013front}
R.~Solimene and A.~Cuccaro, ``Front wall clutter rejection methods in twi,''
  \emph{IEEE Geoscience and remote sensing letters}, vol.~11, no.~6, pp.
  1158--1162, 2013.

\bibitem{vishwakarma2017detection}
S.~Vishwakarma and S.~S. Ram, ``Detection of multiple movers based on single
  channel source separation of their micro-dopplers,'' \emph{IEEE Transactions
  on Aerospace and Electronic Systems}, vol.~54, no.~1, pp. 159--169, 2017.

\bibitem{vishwakarma2020mitigation}
------, ``Mitigation of through-wall distortions of frontal radar images using
  denoising autoencoders,'' in \emph{IEEE Transactions on Geoscience and Remote
  Sensing}.\hskip 1em plus 0.5em minus 0.4em\relax IEEE, 2020.

\bibitem{ram2020estimating}
S.~S. Ram, G.~Singh, and G.~Ghatak, ``Estimating radar detection coverage
  probability of targets in a cluttered environment using stochastic
  geometry,'' in \emph{2020 IEEE International Radar Conference (RADAR)}.\hskip
  1em plus 0.5em minus 0.4em\relax IEEE, 2020, pp. 665--670.

\bibitem{goldstein1973false}
G.~Goldstein, ``False-alarm regulation in log-normal and weibull clutter,''
  \emph{IEEE Transactions on Aerospace and Electronic Systems}, no.~1, pp.
  84--92, 1973.

\bibitem{schleher1976radar}
D.~Schleher, ``Radar detection in weibull clutter,'' \emph{IEEE Transactions on
  Aerospace and Electronic Systems}, no.~6, pp. 736--743, 1976.

\bibitem{conte2004statistical}
E.~Conte, A.~De~Maio, and C.~Galdi, ``Statistical analysis of real clutter at
  different range resolutions,'' \emph{IEEE Transactions on Aerospace and
  Electronic Systems}, vol.~40, no.~3, pp. 903--918, 2004.

\bibitem{haenggi2012stochastic}
M.~Haenggi, \emph{Stochastic geometry for wireless networks}.\hskip 1em plus
  0.5em minus 0.4em\relax Cambridge University Press, 2012.

\bibitem{ebdon1985statistics}
D.~Ebdon, ``Statistics in geography second edition: A practical approach,''
  \emph{Malden, MA: Blackwell Publishing}, 1985.

\bibitem{bufler2016radar}
T.~D. Bufler, ``Radar signature analysis of indoor clutter and stationary human
  target classification,'' 2016.

\bibitem{ram2009radar}
S.~S. Ram, ``Radar simulation of human activities in non line-of-sight
  environments,'' Ph.D. dissertation, University of Texas at Austin, 2009.

\bibitem{ram2010simulation}
S.~S. Ram, C.~Christianson, Y.~Kim, and H.~Ling, ``Simulation and analysis of
  human micro-dopplers in through-wall environments,'' \emph{IEEE Transactions
  on Geoscience and remote sensing}, vol.~48, no.~4, pp. 2015--2023, 2010.

\bibitem{ram2016through}
S.~S. Ram and A.~Majumdar, ``Through-wall propagation effects on
  doppler-enhanced frontal radar images of humans,'' in \emph{2016 IEEE Radar
  Conference (RadarConf)}.\hskip 1em plus 0.5em minus 0.4em\relax IEEE, 2016,
  pp. 1--6.

\bibitem{vishwakarma2020micro}
S.~Vishwakarma, A.~Rafiq, and S.~S. Ram, ``Micro-doppler signatures of dynamic
  humans from around the corner radar,'' in \emph{2020 IEEE International Radar
  Conference (RADAR)}.\hskip 1em plus 0.5em minus 0.4em\relax IEEE, 2020, pp.
  169--174.

\bibitem{ruck1970radar}
G.~Ruck, \emph{Radar Cross Section Handbook: Volume 1}.\hskip 1em plus 0.5em
  minus 0.4em\relax Springer, 1970, vol.~1.

\end{thebibliography}
\end{document}